\documentclass[sigconf]{acmart}

\pdfoutput=1 %

\usepackage{booktabs}
\usepackage{natbib}
\usepackage[utf8]{inputenc}
\usepackage{todonotes}
\usepackage{xspace}
\usepackage[inline]{enumitem}
\mathchardef\mhyphen="2D
\usepackage{macros}

\usepackage{caption}
\usepackage{subcaption}

\graphicspath{{./gfx/}}

\title{The Complexity of Conjunctive Queries with Degree 2}

\author{Matthias Lanzinger}
\orcid{0000-0002-7601-3727}
\affiliation{%
  \institution{Department of Computer Science, University of Oxford}
  \country{United Kingdom}
}
\email{matthias.lanzinger@cs.ox.ac.uk}

\copyrightyear{2022} 
\acmYear{2022} 
\setcopyright{acmlicensed}\acmConference[PODS '22]{Proceedings of the 41st ACM SIGMOD-SIGACT-SIGAI Symposium on Principles of Database Systems}{June 12--17, 2022}{Philadelphia, PA, USA}
\acmBooktitle{Proceedings of the 41st ACM SIGMOD-SIGACT-SIGAI Symposium on Principles of Database Systems (PODS '22), June 12--17, 2022, Philadelphia, PA, USA}
\acmPrice{15.00}
\acmDOI{10.1145/3517804.3524152}
\acmISBN{978-1-4503-9260-0/22/06}



\begin{abstract}
  It is well known that the tractability of conjunctive query answering can be characterised in terms of treewidth when the problem is restricted to queries of bounded arity. We show that a similar characterisation also exists for classes of queries with unbounded arity and degree 2. To do so we introduce hypergraph dilutions as an alternative method to primal graph minors for studying substructures of hypergraphs.  Using dilutions we observe an analogue to the Excluded Grid Theorem for degree 2 hypergraphs. In consequence, we show that that the tractability of conjunctive query answering can be characterised in terms of generalised hypertree width. A similar characterisation is also shown for the corresponding counting problem.  We also generalise our main structural result to arbitrary bounded degree and discuss possible paths towards a characterisation of tractable conjunctive query answering for the bounded degree case.
\end{abstract}

\copyrightyear{2022}
\acmYear{2022}
\setcopyright{acmlicensed}\acmConference[PODS '22]{Proceedings of the 41st
ACM SIGMOD-SIGACT-SIGAI Symposium on Principles of Database
Systems}{June 12--17, 2022}{Philadelphia, PA, USA}
\acmBooktitle{Proceedings of the 41st ACM SIGMOD-SIGACT-SIGAI Symposium
on Principles of Database Systems (PODS '22), June 12--17, 2022, Philadelphia,
PA, USA}
\acmPrice{15.00}
\acmDOI{10.1145/3517804.3524152}
\acmISBN{978-1-4503-9260-0/22/06}

\begin{document}
\fancyhead{}

\maketitle

\section{Introduction}
The complexity of answering conjunctive queries (CQs) has been a classic
topic of study in database theory. CQs make up the core of many common
query languages, such as SQL, SPARQL, or Datalog, and the algorithmic properties of CQs are therefore also critical to query answering in these languages.
Beyond query answering, the complexity of CQs
is of interest throughout theoretical computer science where it is studied extensively under the equivalent frameworks of Constraint
Satisfaction Problems or homomorphisms between relational structures.

When we speak of the complexity of answering CQs, we generally refer to
the decision problem \cqprob, where a CQ $q$ and a database $D$ are
given, and  the task is to decide whether $q$ has a non-empty
set of results when evaluated over the database $D$. In general, \cqprob is \np-complete~\cite{DBLP:conf/stoc/ChandraM77}, but extensive research in the area has yielded large
tractable\footnote{When not stated otherwise we use tractability to mean polynomial-time decidability.}  fragments of the problem by restricting the structure of queries~\cite{DBLP:journals/jcss/GottlobLS02,DBLP:journals/talg/GroheM14}. This
line of study has also produced two important characterisations (in terms of query structure) of tractable CQ
answering. \citet{DBLP:journals/jacm/Grohe07} showed that \cqprob restricted to
bounded arity CQs is tractable exactly for query classes of bounded
\emph{treewidth modulo homomorphism}, i.e., only if there exists some constant $c$ such that every query in the class is equivalent to a query with treewidth at most $c$ (Proposition~\ref{prop:grohehard}). Analogously, \citet{DBLP:journals/jacm/Marx13}
showed that the fixed-parameter tractability of \cqprob
parameterised by the query's hypergraph structure can be characterised
in terms of \emph{submodular width}.

Despite the wide-reaching consequences of these two results, the case of plain
tractability for unbounded arity queries is still not well
understood. While a number of parameters that induce tractable classes
of the problem in the unbounded arity have been identified -- e.g.,
hypertree width~\cite{DBLP:journals/jcss/GottlobLS02} and its
generalisations~\cite{DBLP:journals/jacm/GottlobMS09,DBLP:journals/talg/GroheM14}
-- there is little evidence to suggest whether these parameters
are even close to the limits of tractability, or whether there
exists a \emph{natural} characterisation for the unbounded arity case at all.

What makes the problem challenging is that very little is known of the
hypergraph structure of queries with unbounded hypertree width (or any
related parameters).  Grohe's lower bound critically relies on the
Excluded Grid Theorem
by~\citet{DBLP:journals/jct/RobertsonS86}. Roughly speaking, in this setting the
theorem states that if a query has large treewidth, then its
primal graph will contain a large grid as a graph minor. Intractability of \cqprob can then be shown by reduction of other problems into large enough grids.
However, any minor of the
primal graph lacks crucial information from the query. In particular, it is possible
that large parts of the grid are covered in a single atom and thus the high connectivity of the grid is not reflected in the actual query. In particular, a reduction following Grohe's technique will produce exponentially large relations in such cases and hence not be efficient enough for the hardness results that we are aiming for.

Marx' characterisation in~\cite{DBLP:journals/jacm/Marx13} addresses this issue through the more abstract notion of
\emph{embedding power}. Rather than relying on the existence of
arbitrarily large grid minors, it is shown that in classes of
unbounded submodular width, there  always exist instances with
arbitrarily high embedding power, which in turn allows for ``compact'' embedding of certain other queries. While high embedding power allows
for effective reductions into queries of unbounded arity, it is
not known (nor suspected) that bounded embedding power or submodular
width are sufficient conditions for non-parameterised tractability of \cqprob in the usual setting\footnote{The situation is different when truth-table representation is considered rather than standard ``compact'' representations via lists of tuples. See the discussion of related work on adaptive width below.}.

These observations reveal two important questions in the
search for the limits of tractability for \cqprob when there is no
bound on the arity.
  \emph{
  \begin{enumerate}%
  \item Are there appropriate notions of \emph{forbidden substructures} in hypergraphs of unbounded rank?
  \item Can we relate such forbidden substructures to any common width parameters for hypergraphs?
  \end{enumerate}
  }
\paragraph{Contributions}
In this paper we attempt to answer these questions for hypergraphs with degree 2.
We show that in this setting, large enough generalised hypertree width (\ghw) always implies the existence of certain highly-connected substructures.
This substructure relation, which we call \emph{hypergraph dilution}, is also connected to the complexity of \pcq, the parameterisation of \cqprob by the query. These observations allow us to follow a similar path as Grohe in the proof of the characterisation for bounded arity in~\cite{DBLP:journals/jacm/Grohe07} and obtain a first characterisation result for the complexity of unbounded
arity CQ answering.
\medskip
\begin{center}
  \parbox{0.9\columnwidth}{
    \emph{
      Assume $\wone \neq \fpt$. Let $\mathcal{Q}$ be a class of queries with degree 2 hypergraphs.
      Then $\cqprob(\mathcal{Q})$ is tractable if and only if $\mathcal{Q}$ has bounded semantic generalised hypertree width.
    }
  }
\end{center}
\medskip
  
The main contributions in this paper are summarised as follows.
\begin{enumerate}
\item To capture a type of relevant substructures of hypergraphs, we introduce
  \emph{hypergraph dilutions}  as a possible alternatives to primal graph minors.
 We show that CQ answering over a hypergraph class $\mathcal{M}$ is fpt-reducible to CQ answering over hypergraphs $\mathcal{H}$, if all hypergraphs in $\mathcal{M}$ are dilutions of hypergraphs in $H$.
\item We show an analogue of the Excluded Grid Theorem for degree 2 hypergraphs. In particular, there exists a function $f$ such that for any integer $n>0$, any hypergraph $H$ with $\ghw(H) \geq f(n)$, contains a \emph{jigsaw} hypergraph (the hypergraph dual of a grid) as a hypergraph dilution. This result may also be of independent interest.
\item In consequence, we show that \cqprob over a class of hypergraphs $\mathcal{H}$ is tractable if and only if $\mathcal{H}$ has bounded generalised hypertree width. We extend the result to  classes of queries with bounded semantic generalised hypertree width \cite{DBLP:journals/jacm/BarceloFGP20} and to the corresponding counting problem of counting answers of CQs.
\end{enumerate}
It remains open whether this result can be extended to classes of
arbitrary bounded degree. We propose possible paths to build on the
results presented in this paper to proceed towards this goal. Moreover,
we give a generalisation of the key structural result to the bounded degree case.

\paragraph{Related Work}
To the best of our knowledge, there exists little related previous work on the complexity of \cqprob for unbounded arity or even the structure of hypergraphs of unbounded \emph{rank} (the maximum edge cardinality) beyond the two previously mentioned characterisation results. One important exception is work by~\citet{DBLP:journals/mst/Marx11} which shows that \cqprob is tractable only for classes of bounded adaptive width if the problem is given in truth table encoding (assuming a nonstandard conjecture). Note however that truth-table representation is generally exponentially larger than the standard succinct representation in terms of lists of tuples that we study. 

We study the effect of restricting the query in this paper. This should not be confused with another prominent line of research on tractable fragments arising from restrictions to the structure of the database. There, a full dichotomy theorem is known due to Bulatov and Zhuk~\cite{DBLP:conf/focs/Bulatov17,DBLP:journals/jacm/Zhuk20}. However, the two sides of the problem are completely independent of each other and results for restrictions to the database do not affect the problem discussed here.

It is tempting to ask whether unbounded \ghw also implies \np-hardness of \cqprob in our setting, i.e., whether our main result can be strengthened to a dichotomy. \citet{DBLP:conf/icalp/BodirskyG08} have shown that, in general, no dichotomy for \cqprob exists. That is, there are polynomially constructable classes of CQs for which \cqprob is neither polynomial nor in \np (unless $\ptime=\np$). Moreover, their argument is very flexible and suggests that their result may be extended to hold even under certain structural restrictions to the class of queries (e.g., classes of bounded degree).

\paragraph{Structure}
We continue with preliminary notation and terminology in Section~\ref{sec:prelim}. We introduce hypergraph dilutions and show the fpt-reducibility of CQ answering along hypergraph dilutions in Section~\ref{sec:dilution}. We show the main structural results, and in consequence the complexity lower bounds, for hypergraphs with unbounded \ghw in Section~\ref{sec:forbid}. In Section~\ref{sec:highdeg}, we discuss challenges and possible paths for a characterisation of the bounded degree case. Concluding remarks and directions for further research are discussed in Section~\ref{sec:conc}. Proof details that are skipped in the main body are presented in the appendix.

\section{Preliminaries}
\label{sec:prelim}
For positive integers $n$ we will use $[n]$ as a shorthand for the set $\{1,2,\dots,n\}$. When $X$ is a set of sets
we sometimes write $\bigcup X$ for $\bigcup_{x\in X}x$. We assume the reader to be familiar with standard notions of (parameterised) complexity theory. We refer to~\cite{DBLP:books/daglib/0018514}
and~\cite{DBLP:series/txtcs/FlumG06} for comprehensive overviews of
computational complexity and parameterised complexity, respectively. As usual, we refer to a problem as \emph{tractable} to say that it is in the complexity class \ptime.

\paragraph{Graphs \& Hypergraphs}
A \emph{hypergraph} $H$ is a pair $(V(H), E(H))$ where $V(H)$ is the set of \emph{vertices} and $E(H) \subseteq 2^{V(H)}$
is the set of \emph{(hyper)edges}. We say that an edge $e$ is \emph{incident} to a vertex $v$ if $v \in e$ and refer to the set of all edges incident to $v$ by $I_v$. We treat \emph{graphs} as hypergraphs where every edge has size 2, i.e., 2-uniform hypergraphs.
The \emph{degree} of a vertex $v$ is defined as $\degree(v) := |I_v|$. The degree of a hypergraph is the maximum degree over all its vertices. The \emph{rank} of a hypergraph is $\rank(H):= \max_{e \in E(H)} |e|$. The \emph{primal graph} (or \emph{Gaifman graph}) of a hypergraph $H$ is the graph $G$ with $V(G)=V(H)$ and $\{x,y\} \in E(G)$ if and only if there is some edge in $H$ that contains both $x$ and $y$.

The \emph{dual} $H^d$ of $H$ is the hypergraph with $V(H^d) = E(H)$ and $E(H^d) = \{I_v \mid v \in V(H)\}$.
We say that a hypergraph $H$ is \emph{reduced} if
\begin{enumerate*}
\item every vertex has at least degree 1,
\item $H$ does not contain an empty edge,
\item and no two vertices have the same \emph{vertex type}, i.e., for any two distinct vertices $v,w$, we have $I_v \neq I_w$.
\end{enumerate*}
If a hypergraph is not reduced, we can easily make it reduced by deleting vertices with degree 1, empty edges and all but one vertex for every vertex type. Applying this process to some $H$ yields a \emph{reduced hypergraph} for $H$.
The definition of reduced hypergraphs historically sometimes includes the condition that no two edges are the same. We consider this constraint implicitly always satisfied by our definition of $E(H)$ as a set. Importantly, if $H$ is a reduced hypergraph, then $\left(H^d\right)^d = H$.

A \emph{path} between two distinct vertices $v_0, v_\ell$ in $H$ is a sequence $(v_0, e_0, v_1, e_1,\dots, e_{\ell-1}, v_\ell)$ alternating between vertices $v_i$ and  edges $e_j$ such that $\{v_i, v_{i+1}\} \subseteq e_i$ for all $0 \leq i < \ell$. Furthermore, no edge or vertex occurs twice in a path. 

\emph{Graph minors} will play an important role in this paper. We say that a graph $G$ is a minor of graph $F$ if there exists a function $\mu \colon V(G) \to 2^{V(F)}$  (the \emph{minor map}) such that 
\begin{enumerate}
    \item for every $v\in V(G)$, $\mu(v)$ is connected in $F$,
    \item for any two distinct $v,u \in V(G)$, $\mu(v) \cap \mu(u) = \emptyset$,
    \item and if $v$ and $u$ are adjacent in $G$, then there is an edge in $F$ that connects $\mu(v)$ and $\mu(u)$.
\end{enumerate}
For connected graphs we can assume, w.l.o.g., that a minor map is  \emph{onto}, i.e., $V(F) = \bigcup_{v\in V(G)} \mu(v)$. %
Alternatively, graph minors are also commonly defined constructively in the following way.
An \emph{edge contraction} in a graph $G$ removes an edge $\{v,w\}$ from $G$ and merges the two vertices $v$, $w$ into one new vertex which is adjacent to exactly the edges adjacent to $v$ or $w$, except for the removed $\{v,w\}$.
A graph $G$ is a minor of graph $F$, if $G$ can be reached from $F$ by a sequence of vertex deletions, edge deletions, and edge contractions. 

\paragraph{Width Parameters}
We will be interested in the structure of hypergraphs in the case where certain parameters are large. We follow~\citet{adler2006width} in the following definitions. A tuple $\tdecomp$ is a \emph{tree
  decomposition} of a hypergraph $H$ if $T$ is a tree, every $B_u$
 is a subset of $V(H)$ and the following two conditions are satisfied:
\begin{enumerate*}
\item 
 For every $e \in E(H)$ there
 is a node $u \in T$ s.t. $e \subseteq B_u$, and
\item for every vertex $v \in V(H)$,
  $\{u \in T \mid v \in B_u\}$ is connected in $T$.
\end{enumerate*}
For functions $f\colon 2^{V(H)} \to \mathbb{R}^+$, the
\emph{$f$-width} of a tree decomposition is defined as
$\sup\{f(B_u) \mid u \in T\}$ and the $f$-width of a hypergraph is the
minimal $f$-width over all its tree decompositions.  The
\emph{treewidth} $\tw(H)$ of a hypergraph $H$ is the $w\text{-width}$,
where $w(B)=|B|-1$. An \emph{fractional edge cover} of vertex set  $V' \subseteq V(H)$ is a set of mapping $\gamma : E(H) \to [0,1]$ such that $\sum_{e \in I_v} \gamma(e) \geq 1$ for all $v \in V'$, i.e., $\gamma$ assigns weights to all edges such that every vertex in $V'$ has at least 1 total weight on its incident edges. The weight of a fractional edge cover $\gamma$ is $\sum_{e \in E(H)} \gamma(e)$. The \emph{fractional edge cover number} of $V' \subseteq V(H)$ is the minimum weight of a fractional edge cover of $V'$. An \emph{(integral) edge cover} is a fractional edge cover where every edge is assigned either $0$ or $1$.
Let $\rho$ be the
function associating sets of vertices with their integral edge cover number in
$H$. The \emph{generalised hypertree width} $\ghw(H)$ of $H$ is the
$\rho\text{-width}$. Analogously, one can define
\emph{fractional hypertree width}~\cite{DBLP:journals/talg/GroheM14} as the
$\rho^*$-width where $\rho^*$ is the fractional edge cover
number.
We say that a class of hypergraphs has \emph{bounded} $\ghw$ if there exists a constant $c$ such that for every $H$ in the class, $\ghw(H) \leq c$.
We use the same convention also for other numeric properties of hypergraphs of queries such as degree or treewidth.

The statement of our main result in terms of bounded \ghw may be a source of confusion since
 there exist hypergraph classes with bounded \fhw but unbounded \ghw and bounded \fhw is a sufficient condition for tractability~\cite{DBLP:journals/talg/GroheM14}.
For hypergraphs with bounded degree the two notions are equivalent up to some fixed function, i.e., every class has bounded \fhw if and only if it has bounded \ghw~\cite{JACM21}.
Thus, in the setting considered in this paper we can use the two notions interchangeably.

\paragraph{Conjunctive Queries}

A \emph{conjunctive query} (CQ) $q$ is a function-free conjunction of relational atoms.
Commonly, the definition of CQs also allows for (top-level)
existential quantification of variables. In the context
of this paper, and the decision problem \cqprob as defined below, such quantification is of no consequence and all
results for \cqprob and \pcq hold also with existential
quantification. This is not true for the counting problem where we explicitly consider only \emph{full CQs}, i.e., CQs with no existential quantification. This is discussed further in the
respective Section~\ref{sec:counting}.

A \emph{database}
is a set of ground relational atoms. We say an assignment of $\nu$ of
variables in $q$ to constants is a \emph{solution} of $q$ for $D$ if
every atom in $q$ with variables replaced according to $\nu$ is in
$D$. We denote the set of all solutions of $q$ for $D$ as $q(D)$.
The \emph{arity} of a CQ is the maximal arity of its individual atoms.
If no relation symbol occurs twice in $q$ we say there are no \emph{self-joins}. If $q$ has no self-joins and no repeated variables in any atom we sometimes implicitly treat $q$ as a join-query in relational algebra where the attributes for each relation are simply the lists of variables in the corresponding atoms of $q$.

The hypergraph of $q$ is the hypergraph $H$ such that $V(H)=\vars(q)$ and for every atom $R(x_1,\dots,x_n)$, there exists an edge of the form $\{x_1,\dots,x_n\}$ in $H$ (and no other edges). We transparently refer to properties of the hypergraph of $q$ also as properties as $q$, e.g., by \ghw or degree of $q$ we refer to the \ghw or degree of the hypergraph of $q$.
Throughout this paper we are primarily interested in the following decision problem over some class of CQs $\mathcal{Q}$ known as \emph{Boolean Conjunctive Query Answering}.
\begin{problem}[framed]{$\cqprob(\mathcal{Q})$}
   Instance: & A CQ $q$ in $\mathcal{Q}$ and a database $D$ \\
   Question: & $q(D) \neq \emptyset$?
\end{problem}
We refer to $\cqprob$ parameterised by the hypergraph of the input query $q$ as
$\pcq$. For a hypergraph class $\mathcal{H}$ we write
$\cqprob(\mathcal{H})$ to mean $\cqprob$ over the class of all CQs
whose hypergraph is in $\mathcal{H}$. The same applies to other decision
problems defined over classes of queries.

We say that two CQs $q_1,q_2$ are \emph{equivalent} if $q_1(D) = q_2(D)$ for
every database $D$.  Every CQ $q$ has a minimal (with respect
to the number of atoms) equivalent query which is called the
\emph{core} of $q$, we write $\core(q)$.
For CQ $q$, we will also be interested in the minimal \ghw
over all equivalent  queries. Let $\mathsf{Eq}(q)$ be the equivalence classes of all queries equivalent to $q$. The
\emph{semantic generalised hypertree width} of $q$ ($\sghw(q)$) is $\min\{\ghw(q') \mid q' \in \mathsf{Eq}(q) \}$, i.e., the minimum $\ghw$ in the equivalence class of $q$ . Analogously, we use \emph{semantic treewidth} to refer to the minimum treewidth in the respective equivalence class. Note that semantic width is also commonly referred to as \emph{width modulo homomorphism} in the literature since CQ equivalence coincides with homomorphic equivalence of queries. For full details and related
definitions see~\cite{DBLP:journals/jacm/BarceloFGP20}.

The following two statements for \cqprob will be of particular importance
here. The first is what we informally refer to as Grohe's
characterisation throughout the paper. The second is a straightforward
combination of two standard results of the field, one showing that
$\ghw$ is equivalent to the more restricted notion of \emph{hypertree width} (see~\cite{DBLP:journals/jcss/GottlobLS02}) up to a constant
factor~\cite{DBLP:journals/ejc/AdlerGG07}, and the other showing
tractability of \cqprob under bounded hypertree
width~\cite{DBLP:journals/jcss/GottlobLS02}.
\begin{proposition}[Theorem 1.1, \citet{DBLP:journals/jacm/Grohe07}]
  \label{prop:grohehard}
      \label{prop:grohe}
  Assume $\fpt \neq \wone$. Let $\classQ$ be a recursively enumerable class of bounded arity CQs. The following three statements are equivalent:
  \begin{enumerate}
  \item $\cqprob(\mathcal{Q})$ is tractable;
  \item $\pcq(\mathcal{Q})$ is fixed-parameter tractable;
  \item $\mathcal{Q}$ has bounded semantic treewidth.
  \end{enumerate}
  If either statement is false, then $\pcq(\classQ)$ is \wone-hard.
\end{proposition}
\begin{proposition}[\citet{DBLP:journals/ejc/AdlerGG07,DBLP:journals/jcss/GottlobLS02}]
  \label{prop:hwtrac}
  Let $\mathcal{Q}$ be a class of CQs with bounded \ghw. Then $\cqprob(\mathcal{Q})$ is tractable.
\end{proposition}

\section{Hypergraph Dilutions}
\label{sec:dilution}

In this section, we introduce hypergraph dilutions as a possible
approach to identify relevant substructures of hypergraphs.  As with
graph minors, the goal of this notion is intuitively to induce an
order of structural simplicity in the sense that if $H$ is a
hypergraph dilution of $H'$, then $H$ should be ``simpler'' than
$H'$. The difficulty of course lies in the question of what makes one
hypergraph simpler than another. We do not claim to have an answer to
this question and, moreover, do not propose that there is a single
``correct'' kind of simplicity. Rather, the generality of hypergraphs
suggests that competing notions will be of interest in different
settings. 

In the context of our goal of identifying forbidden substructures for tractable CQ
answering, the desired notion of simplicity is one that captures a
kind of structural abstraction that adheres to a type of monotonicity
of complexity, meaning that \cqprob should not increase in complexity
for simpler, more abstract, structures. Theorem~\ref{thm:hmred} at the
end of this section demonstrates that hypergraph dilutions capture
this high-level idea of structural simplicity and abstraction in a
meaningful way.  In the following section we present further
motivation for the notion, especially for hypergraphs of bounded
degree.

\begin{definition}
  \label{def:dilut}
  For hypergraph $H'$, we say that $H$ is a \emph{hypergraph dilution} of $H'$ if it is isomorphic to a hypergraph that can be reached from $H'$ by a sequence of the following operations:
  \begin{enumerate}
  \item deleting a vertex (from the vertex set and all edges),
  \item deleting an edge that is a proper subset of another edge,
  \item \emph{merging} on $v$: replacing all of the incident edges $I_v$ of vertex $v$, by a new edge $\left(\bigcup I_v\right) \setminus \{v\}$.
  \end{enumerate}
  We also say that $H'$ \emph{dilutes to} $H$ and refer to the associated sequence of operations as a \emph{dilution sequence} from $H'$ to $H$.
\end{definition}

Importantly, hypergraph dilutions do not allow deletion of arbitrary edges. This is
motivated by our interest in the complexity of CQs. Hypergraph
parameters that induce tractable CQ answering usually generalise the
notion of hypergraph $\alpha$-acyclicity~\cite{DBLP:journals/jacm/Fagin83}. An important
observation there is that if there is some complex substructure (say a
clique $C_n$) that is fully contained in a single separate hyperedge,
then the complex interactions of the substructure $C_n$ can roughly
speaking be ignored when solving the associated query. Hence,
removing arbitrary edges can ``activate'' arbitrarily complex
subproblems.

Thus, deleting an edge $e$ is only possible by deleting vertices such
that $e$ becomes a subedge of another (or equal, thus implicitly
disappearing in the other edge). One special case where such deletion
can be convenient is in hypergraphs that are not connected. Having
multiple connected components is technically inconvenient and of
little algorithmic importance -- each component is essentially an
independent instance -- and it is common to assume connected
instances. In the study of hypergraph dilutions this assumption is not
necessary as we can always delete superfluous maximally connected
components by deleting all vertices, leaving only a single empty edge,
which is naturally a proper subset of any other edge.

The following observations on hypergraph dilutions are important in our further studies.
The first two statements of Lemma~\ref{lem:dilut} are straightforward to verify but of technical importance. In particular, the second statement also implies that every hypergraph has only a finite number of dilutions. The third statement is less simple.  Deleting a vertex can
possibly reduce \ghw while deleting (or adding) a subedge cannot
change the width at all. However, the effect of the merging operation
of hypergraph dilutions is less clear since a new large edge is
introduced, forcing vertices to occur in a bag of a decomposition for $H$ that may not occur together in any optimal
decomposition of $H'$. A proof of the third statement is given in the appendix.
\begin{lemma}
  \label{lem:dilut}
  For hypergraphs $H$ and $H'$ such that  $H'$ dilutes to $H$, the following statements hold:
  \begin{enumerate}
  \item $\degree(H) \leq \degree(H')$;
  \item $|V(H)|+|E(H)| < |V(H')|+|E(H')|$;
  \item $\ghw(H) \leq \ghw(H')$.
  \end{enumerate}
\end{lemma}

Our definition of hypergraph dilutions is of course inspired by graph
minors. Previously, \citet{DBLP:journals/tcs/AdlerGK12} introduced the
notion of \emph{hypergraph minors} as an analogue of graph minors for
hypergraphs. There are some important parallels and differences
between hypergraph minors and hypergraph dilutions that merit
discussion. An important concept in hypergraph minors is the contraction of (the primal edge between) two vertices. Informally, contracting two vertices $x,y$ means to replace them by a new vertex $v_{x,y}$ in the vertex set and in all edges that contain either $x$ or $y$. 
\begin{definition}[\citet{DBLP:journals/tcs/AdlerGK12}]
  \label{def:hm}
  For hypergraph $H'$, we say that $H$ is a \emph{hypergraph minor} of $H'$ if $H$ can be obtained  from $H'$ by a sequence of the following operations:
  \begin{enumerate}
  \item deleting a vertex,
  \item deleting an edge that is a proper subset of another edge,
  \item contraction of two vertices that are contained in a common hyperedge,
  \item or adding a hyperedge $e$, if the vertices of $e$ already induce a clique in the primal graph before adding $e$.
  \end{enumerate}
\end{definition}

The main difference between hypergraph minors and dilutions is the difference between contractions and mergings. Figure~\ref{fig:minorex} provides a small illustration of this difference (the merging is on vertex $y$). Not only is the operation different but this simple example already illustrates how dilutions cannot be simulated by hypergraph minors or vice versa. In particular, the contraction in the example creates a vertex with degree 4, increasing the degree of the original graph. Since dilutions can not increase the degree the result of the contraction can not be a dilution of $H$. On the other hand, the merging creates an edge with 4 vertices. The only way this can be achieved using hypergraph minor operations is through adding a new edge over an existing 4-clique in the primal graph. However, there is no way to form a 4-clique in the primal graph, even with contraction.
Furthermore, the last operation in Definition~\ref{def:hm} would be problematic for our reduction in Theorem~\ref{thm:hmred}.
For this reason, we consider only hypergraph dilutions for our structural results in later sections. It remains open whether similar results can be obtained for hypergraph minors. Note also that, in a sense, the contraction operation of hypergraph minors is a dual operation to the edge merging in dilutions. This relationship to contractions in the dual will play an important role in later sections.
\begin{figure}[t]
  \centering
  \includegraphics[width=0.95\columnwidth]{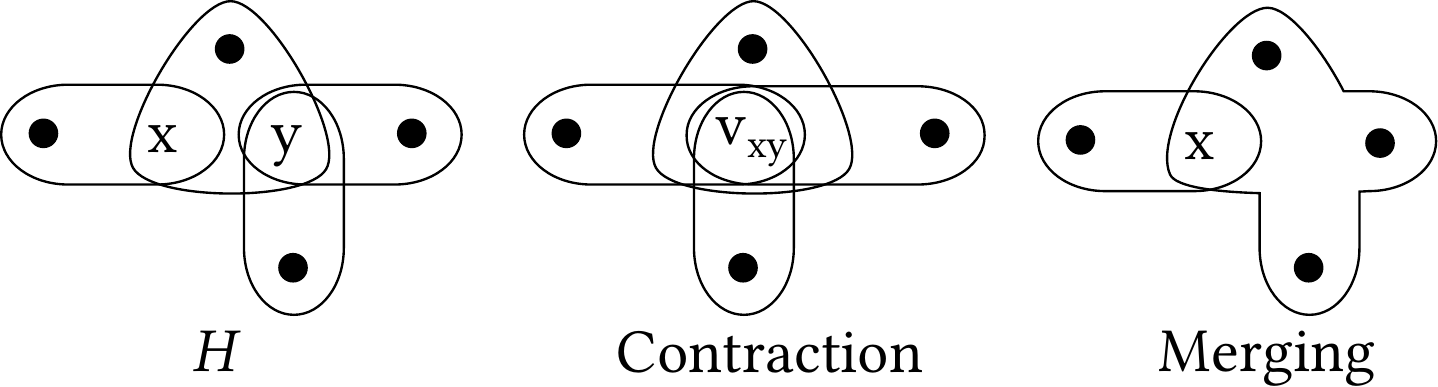}
  \caption{Example: Contraction vs Merging}
  \Description{A visual example of contraction and merging operations on a hypergraph with four edges: $\{a,x\}, \{b,x,y\}, \{y,w\}, \{y,u\}$. The merging is performed on vertex $y$, the contraction contracts vertices $x$ and $y$ into a new vertex $v_{x,y}$}
  \label{fig:minorex}
\end{figure}

Our main results ultimately hinge on two observations. The first is
that for hypergraph classes $\mathcal{H}$ and $\mathcal{M}$, if every hypergraph of $\mathcal{M}$ is a dilution of a hypergraph in $H$, then we can reduce from $\pcq(\mathcal{M})$ to $\pcq(\mathcal{H})$, formally stated in
Theorem~\ref{thm:hmred} below. The second key observation is that
under bounded degree, high \ghw guarantees the existence of certain
dilutions. In combination, these two observations will then yield the lower bounds for our
main results.

 \begin{theorem}
   \label{thm:hmred}
  Let $\mathcal{H}$ be a recursively enumerable class of hypergraphs and let $\mathcal{M}$ be a class such that any member is a hypergraph dilution of a hypergraph in $\mathcal{H}$.
  Then $\pcq(\mathcal{M})$ is fpt-reducible to $\pcq(\mathcal{H})$.
\end{theorem}
\begin{proof}[Proof Idea]
  For some instance $q, D_q$ with hypergraph $M_q$, we find by enumeration of $\mathcal{H}$ a hypergraph $H$ that dilutes to $M_q$ and the corresponding dilution sequence $W=(w_1,\dots,w_\ell)$.
  For each dilution operation $w_i$ -- 
  that produces hypergraph $H_{i}$ from $H_{i-1}$ -- we can show how the query
  $q_{i-1}$ and database $D_{i-1}$ for $H_{i-1}$ can be transformed
  into an equivalent instance $q_i,D_i$ for $H_{i}$ with $\pi_{\vars(q_{i-1})}(q_i(D_i)) = q_{i-1}(D_{i-1})$, where $\pi_{\vars(q_{i-1})}$ is the projection of solutions to the variables of $q_{i-1}$. Thus by traversing $W$ in reverse, we arrive at an instance $p, D_p$ with hypergraph $H_0 = H$ such $\pi_{\vars(p)}(D_p)=q(D_q)$.
  Intuitively, this can be done by introducing keys in the database for the new positions introduced when reversing a merging on a vertex $v$, and by extending all tuples by the same constant to reverse the deletion of a vertex.
  
  For each operation, only linear time in size of the (step $i$) instance is required, and the total size of query and database increases at most in proportion to $\degree(H)$ in each step. Hence, we observe \(\norm{D_p} = O(\degree(H))^\ell \norm{D_q}\) and analogous time bounds for the reduction, where $H$ and $\ell$ both depend only on the parameter $M_q$.
\end{proof}

It may seem natural to extend Definition~\ref{def:dilut} to CQs and
consider reductions from classes of CQ dilutions instead of
operating on hypergraph level. However, it is not clear how the
operations from Definition~\ref{def:dilut} should be adapted to operate
directly on queries. Consider the following example query
$R(x,y,z) \land R(x,u,v) \land S(u,z)$ and consider the case analogous to
deleting vertex $v$ in the corresponding hypergraph. The atom
$R(x,y,z)$ should not be changed but $R(x,u,v)$ would have to become a
$R'(x,u)$ where $R'$ is necessarily a new relation symbol since it has
different arity than $R$. This change in relation symbol removes the
implicit equality between variables $u$ and $y$. It is unclear how the
reduction in Theorem~\ref{thm:hmred} can remain polynomial in the size
of $D$ if such situations occurred. Similar issues can arise when two
edges in the hypergraph are merged into one. Note however that these
problems only arise in the presence of self-joins and that
Theorem~\ref{thm:hmred} can be adapted to hold for classes of
self-join free queries. In Section~\ref{sec:semantic} we discuss how we can still derive our lower bounds for classes of queries through combination with previous results relating the complexity of all queries over a class of hypergraphs to specific classes of queries.

The complexity of deciding hypergraph dilutions is of little consequence to the contents of this paper.
As the complexity may be of independent interest we state it here. An argument is given in the appendix. %
\begin{theorem}
  \label{thm:npdilt}
  It is \np-complete to decide for input hypergraphs $H$ and $H'$, whether $H$ is a hypergraph dilution of $H'$.
\end{theorem}

It is often technically convenient to consider the analogue of reduced hypergraphs for CQs. That is, we want to assume that no variables occur only in one atom, no atom's variables are a subset of some other atom's variables, and so on. These assumptions on CQs are usually motivated by the fact that they have no significant effect on the upper bounds of the problem and can be avoided via straightforward preprocessing. In conjunction with Theorem~\ref{thm:hmred}, the complexity implications of simplifying CQs in this way can be seen via the following Lemma~\ref{lem:redhyper}, which will also be of technical importance in the following section.

\begin{lemma}
  \label{lem:redhyper}
  Let $H$ be a reduced hypergraph for $H'$. Then $H'$ dilutes to $H$, and a corresponding dilution sequence can be computed in polynomial time.
\end{lemma}

\section{Forbidden Dilutions for Degree 2 CQs}
\label{sec:forbid}
  \begin{figure*}[t]
   \centering
   \includegraphics[width=0.95\textwidth]{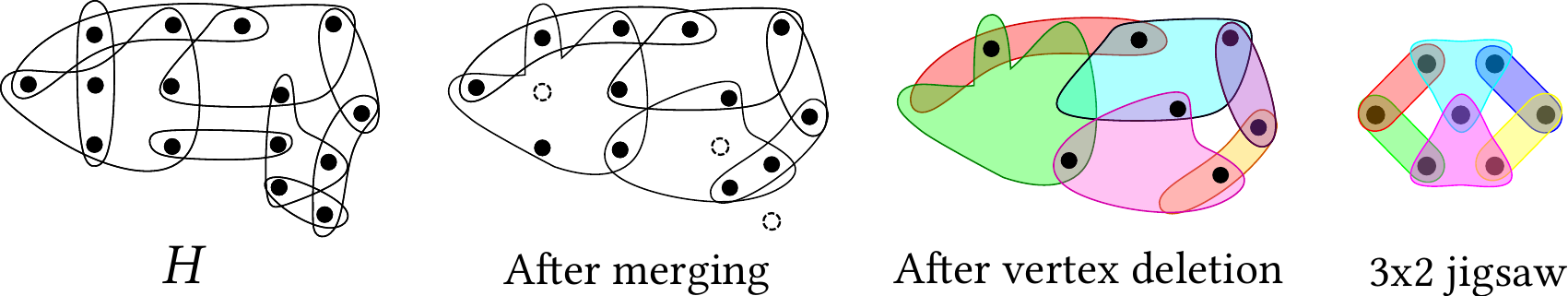}
   \caption{Example Dilution from $H$ to the $3\times{}2$-jigsaw.}
   \Description{An example dilution of a complex hypergraph to a jigsaw.}
   \label{fig:exmesh}
 \end{figure*}

\begin{figure}[t]
  \centering
  \includegraphics[width=0.35\columnwidth]{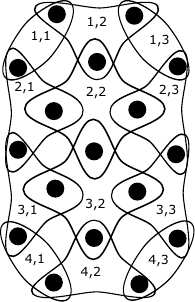}
  \caption{The $3\times{}4$-jigsaw hypergraph}
  \Description{A visual example to of a jigsaw hypergraph to support Definition~\ref{def:jigsaw}.}
  \label{fig:exjig}
\end{figure}

In this section we show that degree 2 hypergraphs with high \ghw always dilute to certain simple but highly connected structures. In particular, we obtain an analogue to the Excluded Grid Theorem for degree 2 hypergraphs. We show that \pcq over these contained structures is hard and thus putting everything together yields the base version of our main result.
\begin{theorem}
  \label{thm:de2}
  \label{thm:deg2}
  Assume $\fpt \neq \wone$. Let $\mathcal{H}$ be a recursively enumerable class of hypergraphs with degree 2. The following three statements are equivalent:
  \begin{enumerate}
  \item $\cqprob(\mathcal{H})$ is tractable; \label{de2trac}
  \item $\pcq(\mathcal{H})$ is fixed-parameter tractable; \label{de2fpt}
  \item $\mathcal{H}$ has bounded generalised hypertree width. \label{de2bound}
  \end{enumerate}
  If either statement is false, then $\pcq(\classH)$ is \wone-hard.
\end{theorem}

In general, there are tractable classes of \cqprob that have bounded fractional hypertree width but unbounded generalised hypertree width. In this light, the characterisation in terms of generalised hypertree width may seem unintuitive. However, for bounded degree \classH (and actually even more general restrictions) it is known that \classH has bounded \fhw if and only if \classH has bounded \ghw~\cite{JACM21}. Theorem~\ref{thm:de2} can therefore equivalently be stated in terms of fractional hypertree width (or just hypertree width).

\subsection{The Structure of Hypergraphs with Degree 2 and Unbounded Generalised Hypertree Width}

We will show that degree 2 hypergraphs always dilute to the hypergraph dual of a grid graph, which we will call a \emph{jigsaw hypergraph}. 
\begin{definition}[Jigsaw Hypergraphs]
  \label{def:jigsaws}
  An \emph{$\nxm$-jigsaw} is a hypergraph $H$ with edges
  $\{e_{i,j} \mid i,j \in [n]\times [m]\}$ where every vertex has
  degree 2 and $|e_{i,j} \cap e_{i+1,j}|=1$ and $|e_{i,j} \cap e_{i,j+1}|=1$ for
   $i<n$, $j<m$  and no other pair of edges has a non-empty intersection. %
\end{definition}

The $\nxm$-jigsaw is uniquely determined up to
isomorphism. Figure~\ref{fig:exjig} illustrates a $3\times{}4$-jigsaw
hypergraph. We call $\nxm$ the \emph{dimension} of the jigsaw and we
say that a class of jigsaws has unbounded dimension if there is no
constant bound on either parameter. Note that the the $\nxm$-jigsaw
dilutes to the $n\times{}(m-1)$ jigsaw (and analogously in the other
axis).

\begin{example}
  \label{ex:dil}
 Figure~\ref{fig:exmesh} illustrates an example dilution of a hypergraph with degree 2 to a to the $3\times{}2$-jigsaw.
  In the first step in the figure, three merging operations are performed. The vertices which we merge on are drawn as dashed empty circles. In a second step we delete superfluous vertices. The colours of the edges represent the correspondence to edges in the final jigsaw. 
\end{example}

Our first goal in this section will be to show that it is always possible to dilute a degree 2 hypergraph $H$ to an $\nxn$-jigsaw where $n$ depends on $\ghw(H)$.
We will first observe that graph minors and hypergraph dilutions are
tightly connected in degree 2 hypergraphs. From there we then derive our
main structural result (Theorem~\ref{thm:excjig}).

\begin{lemma}
  \label{lem:shmin}
  Let $G$ be
  a connected graph and let $H$ be a degree 2 hypergraph. If $G$ is a
  minor of $H^d$, then $G^d$ is a hypergraph dilution of $H$.
\end{lemma}
\begin{proof}
  We assume that $H$ is a reduced hypergraph. Isolated vertices, empty
  edges and duplicate vertex types do not materially affect
  minor maps from $G$ into $H^d$. By Lemma~\ref{lem:redhyper}, there is always
  a dilution sequence from any hypergraph to its respective reduced
  version. Hence, the assumption can be made without loss of
  generality.

  Let $\phi \colon E(H) \to V(H^d)$ be the bijection from edges in $H$ to their corresponding vertex in the dual, and let 
  $\mu \colon V(G) \to 2^{V(H^d)}$ be a minor map from $G$ onto $H^d$. For every $v \in V(G)$, let $\delta(v) = \phi^{-1}(\mu(v))$ and observe that $\delta(v)$ is a connected set of
  edges in $H$.

  For any two adjacent vertices $u,v$ in $G$, there is an edge in
  $H^d$ that connects $\mu(u)$ and $\mu(v)$. Hence, there is also a
  vertex $c_{u,v}$ that is both in an edge in $\delta(u)$ and an edge in $\delta(v)$.  Since
  $c_{u,v}$ has degree 2, it is therefore connected to only one edge
  in $\delta(u)$. For each $v$ adjacent to $u$ in $G$ fix such a $c_{u,v}$ and let us refer to the set of these fixed vertices for $u$ as $C_u$.
  Let $\tau_u$ be the vertices that are incident only to edges in $\delta(u)$. Observe that either $\delta(u)$ conains one edge, or every edge in $\delta(u)$ is incident to at least one vertex $\tau_u$. Suppose towards a contradiction that $\delta(u)$ consists of more than one edge and and that there is an edge $e \in \delta(u)$ such that all vertices in $e$ are incident to some other edge not in $\delta(u)$. Since all vertices have have degree at most 2, that would imply that $e$ is not incident to any other edge in $\delta(u)$, thus contradicting the connectedness of $\delta(u)$.
  Furthermore, note that $\tau_u$ and $C_u$ are disjoint by definition.

  Let $H_1$ be the hypergraph
  obtained by merging, for every $u\in V(G)$,  all vertices in $\tau_u$. By the above observations, either $\delta(u)$ was already a singleton, or the merging produced a single new merged edge $e_u$ from all of the edges of $\delta(u)$, since every such edge was incident to some vertex in $\tau_u$.  By  construction, $H_1$ is clearly a dilution of $H$. Let
  $C = \bigcup_{u \in V(G)}C_u$  and observe that since
  $\tau_u \cap C_u = \emptyset$ for all $u\in V(G)$, no vertices in a
  $C$ have been removed by the merging process.

  Finally, let $H_2$ be the induced subhypergraph $H_1[C]$, i.e., the
  hypergraph obtained from $H_1$ by deleting all vertices not in $C$.
  Observe that for every edge $u \in E(G^d)$, there is a vertex
  $u \in V(G)$ and exactly one edge $e_u \cap C$ in $H_2$. For every edge $\{u,v\}$ in $G$ (or vertex $g_{u,v}$ in $G^d$), there is a vertex $c_{u,v}$ in $C$ and thus in $H_2$, such that $c_{u,v}$ is contained only in edges $e_u$ and $e_v$. Since this correspondence from edges and vertices of $G^d$ to vertices and edges in $H_2$ is one-to-one and $H_2$ contains only these edges and vertices by construction, the implications hold also in the other direction. Hence, $H_2$ is isomorphic to $G^d$ and a hypergraph dilution of $H$.
\end{proof}

This observed duality of graph minors and dilutions in degree 2 hypergraphs also illustrates a conceptual switch. Intuitively, high treewidth expresses large sets of highly connected vertices, while high \ghw can be seen as a sign of large sets of highly connected edges. See also the discussion accompanying the definition of embedding power in~\cite{DBLP:journals/jacm/Marx13} for further intuition.

\begin{proposition}[\citet{DBLP:journals/jct/RobertsonS86}]
  \label{prop:excgrid}
  There exists a function $f \colon \mathbb{N} \to \mathbb{N}$ with the following property: for every $n\geq 1$,
  every graph $G$ with $\tw(G) > f(n)$ contains an $\nxn$-grid as a minor.
\end{proposition}

As a final piece of the puzzle we observe that high \ghw always implies high treewidth in the dual. This observation has been informally mentioned previously, but we are not aware of any formal statement or proof in the literature. Since it is key to our main theorem we provide our own proof in the appendix.
\begin{lemma}
  \label{lem:dualtw}
  Let $H$ be a reduced hypergraph. Then $\ghw(H) \leq tw(H^d)+1$.
\end{lemma}

\begin{theorem}
  \label{thm:excjig}
  There exists a function $f$ with the following property: for every $n\geq 1$,
  every degree 2 hypergraph $H$ with $\ghw(H) > f(n)$ dilutes to the
  $\nxn$-jigsaw.\footnote{The upper bound for $f$ is inherited from the Grid Exclusion Theorem. The best known bound currently is $f(n) = O(n^9 \polylog n)$ due to~\citet{DBLP:journals/jctb/ChuzhoyT21}}
\end{theorem}
\begin{proof}
  Let $r \colon \mathbb{N} \to \mathbb{N}$ be the function from
  Theorem~\ref{prop:excgrid}. For the function of the statement it suffices to consider $f \colon n \mapsto r(n)+1$. Let $H'$ be a hypergraph with
  $\ghw(H') > f(n)$, let $H$ be the reduced hypergraph for $H'$ and
  recall that $\ghw(H)=\ghw(H')$.  By Lemma~\ref{lem:dualtw}, we have
  that $tw(H^d) > f(n)-1 = r(n)$ and thus $H^d$ contains a $\nxn$-grid $G_n$
  as minor.  By Lemma~\ref{lem:shmin}, $G_n^d$ is a hypergraph
  dilution of $H$ and by Lemma~\ref{lem:redhyper} also of $H'$. By
  definition $G_n^d$ is the $\nxn$-jigsaw  $J_n$ and thus
  $J_n$ is a hypergraph dilution of $H'$.
\end{proof}

\subsection{From Jigsaw Dilutions to Lower Bounds}

It is not difficult to observe that the $\nxn$-jigsaw has $\ghw$ of at
least $n$. This can be seen by observing that since the jigsaw can not be separated by less
than $n$ edges it can not be separated into balanced components (that is, components at most half the size of the original hypergraph) by less than $n$
edges. It is known that such balanced separation of a hypergraph $H$ can always be achieved with $\ghw(H)$ edges~\cite{DBLP:journals/ejc/AdlerGG07} and hence $\ghw$ of the $\nxn$-jigsaw must be at least $n$.
Moreover, from Lemma~\ref{lem:dilut} we can also observe the opposite direction: a hypergraph $H$ has high \ghw if it dilutes to a jigsaw with high dimension (regardless of the degree of $H$).

We are now ready to combine our main structural results with our
reduction for dilutions to derive our lower bound for degree 2 CQ answering.

\begin{theorem}
  \label{thm:mainhard}
  Let $\mathcal{H}$ be a recursively enumerable class of degree 2
  hypergraphs with unbounded $\ghw$. Then $\pcq(\mathcal{H})$ is
  $\wone$-hard under fpt-reductions.
\end{theorem}
\begin{proof}
  First we observe that if $\mathcal{J}$ is a recursively enumerable class of jigsaws with unbounded
  dimension, then $\pcq(\mathcal{J})$ is $\wone$-hard.
    From the above discussion $\mathcal{J}$ has unbounded \ghw and thus
  also unbounded treewidth.  Let
  $\mathcal{Q}_{\mathcal{J}}$ be the class of all self-join free
  queries with no repeat variables in any atom and hypergraphs in
  $\mathcal{J}$. $\mathcal{Q}_{\mathcal{J}}$ then has arity 4 and unbounded
  semantic treewidth and $\pcq(\mathcal{Q}_{\mathcal{J}})$ is \wone-hard under fpt-reductions by
  Proposition~\ref{prop:grohehard}. Then, by inclusion so is $\pcq(\mathcal{J})$.

  Let $\mathcal{H}^*$ be the class of all dilutions of
  $\mathcal{H}$. Note that $\mathcal{H}^*$ can still be recursively
  enumerated. By Theorem~\ref{thm:excjig}, and the previous observation that an $\nxn$-jigsaw dilutes to all (modulo isomorphism) jigsaws of lower dimension, $\mathcal{H}^*$ contains
  the class of all jigsaws and thus $\pcq(\mathcal{H}^*)$ is
  \wone-hard by the argument above. Then by Theorem~\ref{thm:hmred} so
  is $\pcq(\mathcal{H})$.
\end{proof}

\begin{proof}[Proof of Theorem~\ref{thm:de2}]
  The implication \ref{de2bound}$\Rightarrow$\ref{de2trac} follows directly
  from Proposition~\ref{prop:hwtrac}. The implication \ref{de2trac}$\Rightarrow$\ref{de2fpt} is immediate.  If $\mathcal{H}$ has unbounded $\ghw$,
  then $\pcq(\mathcal{H})$ is \wone-hard by
  Theorem~\ref{thm:mainhard}. Since we assume that $\fpt \neq \wone$,
  the implication \ref{de2fpt}$\Rightarrow$\ref{de2bound} follows by
  contraposition.
\end{proof}

Theorem~\ref{thm:deg2} also has interesting structural consequences.
According to~\citet{DBLP:journals/jacm/Marx13},
$\pcq(\classH)$ is fixed-parameter tractable if and only if $\classH$
has bounded submodular width (\subw), assuming the Exponential Time
Hypothesis (ETH)~\cite{DBLP:journals/jcss/ImpagliazzoPZ01}. Recall, the ETH
is a stronger assumption than $\fpt \neq \wone$ in the sense that, if the
ETH holds, so does $\fpt \neq \wone$. It holds for any hypergraph $H$ that $\subw(H) \leq \ghw(H)$, but the two complexity results
imply a previously unknown, and somewhat surprising,
equivalence of the two width parameters for degree 2 hypergraphs.

\begin{corollary}
  \label{cor:subw}
  Assume the Exponential Time Hypothesis. Let $\mathcal{H}$ be a recursively enumerable class of degree 2 hypergraphs. Then
  $\mathcal{H}$ has bounded submodular width if and only if it has
  bounded generalised hypertree width.
\end{corollary}

Finding a constructive argument for Corollary~\ref{cor:subw} is an
interesting open question in the search for further lower bounds
beyond degree 2. We refer to Section~\ref{sec:highdeg} for further
discussion.

\subsection{To Classes of Queries}
\label{sec:semantic}

We can make Theorem~\ref{thm:de2} more fine-grained. Instead of all
queries for a class of hypergraphs we can also consider just classes
of queries as in Proposition~\ref{prop:grohe}.  See
also~\cite{DBLP:conf/ijcai/ChenGLP20} for the respective extension to
Marx' characterisation of fixed-parameter tractability and further
discussion of the differences.

As discussed above, it is not clear how to handle hypergraph dilutions
on a query level. Consequently, it is also difficult to state an
analogue to the reduction in Theorem~\ref{thm:hmred} for classes of
queries.  Instead, we can make use of a more general result by~\citet{DBLP:conf/ijcai/ChenGLP20}
that relates the complexity of CQ answering over classes of
hypergraphs to the complexity of query classes.

\begin{proposition}[\citet{DBLP:conf/ijcai/ChenGLP20}]
  \label{prop:hgred}
  Let $\classQ$ be a class of CQs, let $\core(\classQ)$ be the class of cores of \classQ and let $\classH^{\core(\classQ)}$ be the class of hypergraphs of the queries in $\core(\classQ)$. Then  $\pcq(\classH^{\core(\classQ)})$ is fpt-reducible to $\pcq(\classQ)$.
\end{proposition}

There is some ambiguity in what can be considered a degree 2 CQ. The
hypergraph of a query can have degree 2 even if variables occur in
more than 2 atoms of a query. For example, in the query
$R(x,y) \land S(x,y) \land T(x,z)$, $x$ is in 3 atoms but only in two
edges of the hypergraph since the $R$ and $S$ atoms become the same
edge. The following results hold also for the more expansive reading,
that is, we say that a CQ has degree 2 if its hypergraph has degree 2.

\begin{theorem}
  \label{thm:hardq}
  Assume $\fpt \neq \wone$. Let \classQ be a recursively enumerable
  class of degree 2 CQs that does not have bounded semantic generalised
  hypertree width. Then \pcq(\classQ) is \wone-hard.
\end{theorem}
\begin{proof}
  Let $\mathcal{H}^{\core(\classQ)}$ be the class of all hypergraphs
  of the cores of the queries in $\classQ$.  It is known that the
  semantic generalised hypertree width $\sghw(q)$ of a CQ $q$ is
  precisely
  $\ghw(\core(q))$~\cite{DBLP:journals/jacm/BarceloFGP20}. Thus, if
  \classQ has unbounded \sghw, $\mathcal{H}^{\core(\classQ)}$ has
  unbounded \ghw. Recall that the hypergraph of $\core(q)$ is a
  subhypergraph of the hypergraph of $q$ and thus will also have
  degree 2. Thus, we can apply Theorem~\ref{thm:deg2} and see that
  $\pcq(\mathcal{H}^{\core(\classQ)})$ is \wone-hard. By
  Proposition~\ref{prop:hgred} the same also holds for $\pcq(\classQ)$.
\end{proof}
Note that semantic fractional hypertree width is also equal to the
\fhw of the core~\cite{DBLP:conf/ijcai/ChenGLP20} and thus again
bounded if and only if \sghw is bounded, assuming bounded degree.

The tractability of $\cqprob(\classQ)$ where $\classQ$ has bounded
\sghw is known due to~\citet{DBLP:conf/cp/ChenD05}. Thus by analogous
argument to Theorem~\ref{thm:de2} we also observe the following extension.
\begin{theorem}
  \label{thm:semde2}
    Assume $\fpt \neq \wone$. Let $\mathcal{Q}$ be a class of degree 2 CQs. The following three statements are equivalent:
  \begin{enumerate}
  \item $\cqprob(\mathcal{Q})$ is tractable;
  \item $\pcq(\mathcal{Q})$ is fixed-parameter tractable;
  \item $\mathcal{Q}$ has bounded semantic generalised hypertree width.
  \end{enumerate}
\end{theorem}

\subsection{Counting}
\label{sec:counting}

\citet{DBLP:journals/tcs/DalmauJ04} showed a matching result to
Proposition~\ref{prop:grohe} for to the corresponding counting problem
\cqcount.  To be precise, by \cqcount we consider the problem of
computing $|q(D)|$ for given \emph{full} CQ $q$  and database $D$. We also again consider the
parameterisation by the query hypergraph \pcqcount. In this setting, The main result of~\cite{DBLP:journals/tcs/DalmauJ04} then reads as follows.
\begin{proposition}[\citet{DBLP:journals/tcs/DalmauJ04}]
    \label{prop:dalmau}
    Assume $\fpt \neq \sharpwone$\footnote{By slight abuse of notation we also refer to the class of fixed-parameter polynomial counting problems as $\fpt$ when speaking of counting problems}. Then for every recursively enumerable class $\mathcal{Q}$ bounded arity CQs the following three statements are equivalent.
    \begin{enumerate}
    \item $\cqcount(\mathcal{Q})$ is in $\fptime$;
    \item $\pcqcount(\mathcal{Q})$ is in $\fpt$;
    \item $\mathcal{Q}$ has bounded treewidth.
    \end{enumerate}
\end{proposition}

Recall that we only consider full CQs, i.e., queries with
no existential quantification. For counting this is an important
restriction since~\citet{DBLP:journals/jcss/PichlerS13} show that
even for acyclic CQs the problem is \sharpp-complete in the presence
of even a single existentially quantified variable. This restriction
also aligns our problem \cqcount with the popular problem of counting
homomorphisms when viewing $q$ and $D$ as relational
structures. \citet{DBLP:journals/jcss/PichlerS13} also establish the
following upper bound.

\begin{proposition}[\citet{DBLP:journals/jcss/PichlerS13}]
  \label{prop:counttrac}
  Let $\mathcal{Q}$ be a class of CQs with no existential quantification and bounded \ghw. Then $\cqcount(\mathcal{Q})$ is in $\fptime$.
\end{proposition}

Recall the reduction from Theorem~\ref{thm:hmred}. In the full proof we
show that, modulo projection, the result of the reduction produces
the exact same results as the original query.  Through further
inspection of the full proof it is not difficult to verify that even without projection
the number of solutions stays the exact same after the reduction, i.e., the reduction is \emph{parsimonious} (cf.,~\cite{DBLP:journals/siamcomp/FlumG04}).

 \begin{theorem}
   \label{thm:countred}
  Let $\mathcal{H}$ be a recursively enumerable class of hypergraphs and let $\mathcal{M}$ be a class such that any member is a hypergraph dilution of a hypergraph in $\mathcal{H}$.
  Then $\pcqcount(\mathcal{M})$ is fixed-parameter parsimonious reducible to $\pcqcount(\mathcal{H})$.
\end{theorem}

From Proposition~\ref{prop:dalmau} it is straightforward to derive an
analogue of Theorem~\ref{thm:mainhard} for $\pcqcount$. Combining this observation
with Theorem~\ref{thm:countred} and Proposition~\ref{prop:counttrac} we
can then also obtain the matching result for the counting problem for full CQs with degree 2 and unbounded arity. In the following we write $\cqcount(\mathcal{H})$ and $\pcqcount(\mathcal{H})$, where $\mathcal{H}$ is a class of hypergraphs, for the problems $\cqcount$ and $\pcqcount$, respectively, restricted to all full CQs with hypergraph in $\mathcal{H}$.
\begin{theorem}
  \label{thm:count}
    Assume $\fpt \neq \sharpwone$. Then for every recursively enumerable class $\mathcal{H}$ of degree 2 hypergraphs following three statements are equivalent.
    \begin{enumerate}
    \item $\cqcount(\mathcal{H})$ is in $\fptime$;
    \item $\pcqcount(\mathcal{H})$ is in $\fpt$;
    \item $\mathcal{H}$ has bounded generalised hypertree width.
    \end{enumerate}
\end{theorem}

\section{On Arbitrary Bounded Degree}
\label{sec:highdeg}

The results of the previous section ask a natural next question: what about arbitrary bounded degree?
In this section we briefly
discuss possible paths towards this goal and give a generalisation of our main structural result to arbitrary fixed degrees.

It is an open question whether Theorem~\ref{thm:excjig} holds also under the presence of bounded degree above 2. We can however state the analogous theorem for a generalisation of jigsaw hypergraphs that we will call pre-jigsaws.
\begin{definition}
  \label{def:prej}
  Let $J$ be an $\nxm$-jigsaw and $H$ a hypergraph. We say $H$ is a \emph{$\nxm$-pre-jigsaw}
  if there is an mapping $\pi \colon V(J) \to V(H)$ and a mapping $o \colon E(J) \to 2^{E(H)}$ such that:
  \begin{enumerate}
  \item for every two edges $e,f \in E(J)$, $o(e) \cap o(f) = \emptyset$,
  \item every edge in $H$ is in one image $o(e)$ for some $e\in E(J)$,
  \item for two vertices $u,v$ in the same edge $e$ of $J$, we can fix a path $P_{u,v}$ from $\pi(u)$ to $\pi(v)$ using only edges in $o(e)$ and no vertices in the image of $\pi$ other than $\pi(u)$ and $\pi(v)$,\label{paths}
  \item and every vertex in $V(H)$ is either in the image of $\pi$, or occurs in on of the fixed paths of Property~\ref{paths}.
  \end{enumerate}
\end{definition}

Pre-jigsaws generalise jigsaws in the sense that each single edge $e$
of a jigsaw is replaced by paths between the four
vertices in $e$. Moreover, this ``internal'' connection of vertices by a
jigsaw edge $e$ is replaced only by paths using the edges in $o(e)$. %
Note also that a jigsaw is
also a pre-jigsaw and every degree 2 $\nxm$-pre-jigsaw dilutes to a
$\nxm$ jigsaw by merging on the vertices in connecting paths from point \ref{paths} of Definition~\ref{def:prej}.

However, to obtain Theorem~\ref{thm:prej}, our definition of pre-jigsaws makes an important compromise.
While the path $P_{u,v}$ for $u,v$ in $e$ from the definition uses only edges in $o(e)$, it is still possible that an edge $f \in E(H)$ with $f \not \in o(e)$ contains a vertex $w$ that is used in the path $P_{u,v}$. This possibility of edges touching other paths is the key technical differences between  jigsaws and pre-jigsaws. The merging along the connecting paths to obtain a $\nxm$-jigsaw from a degree 2 $\nxm$-pre-jigsaw noted above is not always possible when the pre-jigsaw has degree greater than 2. Merging on the vertex $w$ in path $P_{u,v}$ and edge $f$ from above would merge edges in $o(e)$ with the edge $f \not\in o(e)$, and the resulting hypergraph after merging along paths will not be a jigsaw. Moreover, such edges that touch other paths can also be a source of unbounded arity, which in turn makes it unlikely that we can use Proposition~\ref{prop:grohe} directly to derive hardness for important classes of pre-jigsaws. However, even in extreme cases, the structure of pre-jigsaws is not trivial and the fact that certain hypergraphs always dilute to large pre-jigsaws is still significant. 

The critical Lemma~\ref{lem:shmin} from the degree 2 case does not
hold for higher degrees. Through a similar, but much more
involved, argument over the dual hypergraph one can still show that
high treewidth in the dual hypergraph implies the existence of a large pre-jigsaw. A full proof and further details are available in the extended version of this paper~\cite{DBLP:journals/corr/abs-2111-11532}.
\begin{theorem}
  \label{thm:prej}
  For every $d\geq 1$, there exists a function $f_d  \colon \mathbb{N} \to \mathbb{N}$ with the following property: for every
  $n\geq 1$, every hypergraph $H$ with degree $d$ and $\ghw(H) > f_d(n)$
  dilutes to an $\nxn$-pre-jigsaw.
\end{theorem}

With respect to finding a characterisation of tractability for the bounded degree case, Theorem~\ref{thm:prej} is only a first step. In general, a hypergraph class $\mathcal{H}$ with unbounded \ghw and bounded degree may not contain all pre-jigsaws as dilutions of its members, but only some pre-jigsaws (cf. the proof of Theorem~\ref{thm:mainhard}). Recall that the \nxn-jigsaw dilutes to all lower dimension jigsaws, and therefore a class with degree 2 and unbounded \ghw will contain all jigsaws as its dilutions.  The same does not hold for pre-jigsaws, introducing further complexity to the bounded degree case. It is therefore of interest whether Theorem~\ref{thm:prej} can be made more precise in terms of showing that specific kinds of pre-jigsaws always exist as dilutions of hypergraphs with high \ghw.

Further exploration of Corollary~\ref{cor:subw} may offer an alternative path to the desired result. While the corollary states that submodular width and generalised hypertree width are equivalent under degree 2, the result is observed as a consequence of our complexity results and it remains unclear how to show the equivalence from a structural perspective. A structural argument would likely provide important further insight in the interaction between the two width parameters and may be amenable to a generalisation to bounded degree.

\section{Conclusion \& Outlook}
\label{sec:conc}
We have proposed hypergraph dilutions as an alternative to graph
minors in the study of structural properties of hypergraphs. While the two notions are connected technically, dilutions operate on the hypergraph level and therefore avoid critical issues with graph minors in the presence of arbitrarily large hyperedges. Our study of dilutions yields analogues of the Excluded Grid
Theorem and Grohe's characterisation of tractability for bounded arity
CQ answering, for degree 2 hypergraphs.  To the best of our knowledge
these are the first such results for hypergraphs of unbounded rank.

It remains open whether such a neat
delineation of tractable CQ answering even exists under more general
circumstances such as bounded degree. In support of this natural next step, we show a generalisation of our main structural result for fixed degree and discuss possible paths to extend the presented results to bounded degree. As an immediate next goal we hope to find a more informative proof of Corollary~\ref{cor:subw}, with the eventual goal of better understanding the submodular width of unbounded pre-jigsaws.

Dilutions are closely related to graph minors and our results here rely on key results for graph minors. However, recent thought in graph theory has identified \emph{tangles} as possibly even more
fundamental notion of what it means for a graph to be highly connected (e.g., see the discussion in~\cite{DBLP:journals/jct/RobertsonS03a}).
\citet{DBLP:journals/ejc/AdlerGG07} have previously generalised tangles to \emph{hypertangles} and showed their connection to other hypergraph notions (such as \ghw). The further study of tangles in hypergraphs thus presents an interesting alternative direction towards further understanding substructures in hypergraphs.

Finally, we are not aware of a version of Proposition~\ref{prop:hgred} for counting, and it is not immediate whether the arguments apply also for counting problems. Extending Theorem~\ref{thm:count} to classes of queries is left as an open problem. Recently, it has been shown that \cqcount is also difficult to approximate~\cite{DBLP:journals/toct/BulatovZ20} under certain conditions. Whether the more elaborate machinery for the approximation case also translates to our setting is a further interesting open question.

\section*{Acknowledgements}
Matthias Lanzinger acknowledges support by the Royal Society  project "RAISON DATA" (Project reference: RP\textbackslash R1\textbackslash 201074).
The author is grateful to the detailed feedback by anonymous referees which has greatly improved this manuscript. The author would also like to thank Marco Bressan for reporting an error in an earlier draft.

\bibliographystyle{ACM-Reference-Format}
\bibliography{refs}

\appendix
\section{Degree 2 in Practice}
While the contributions of this paper are primarily theoretical and the degree 2 case is viewed as a first step towards possible broader characterisations, it may
be of interest how relevant the degree 2 case is in practice. This is of particular relevance as degree 2 in graphs is highly restrictive, with only line graphs and cycles satisfying the condition. In hypergraphs the situation is different, and much more complex structures can be constructed with degree 2 as was already shown through $\nxn$-jigsaws or the example in Figure~\ref{fig:exmesh}. 

To offer some further perspective on this question we present some statistics from the HyperBench~\cite{10.1145/3440015} benchmark. HyperBench consists of collection of hypergraphs from synthetic and real-world CQs and Constraint Satisfaction Problems.

Of the 3649 hypergraphs in HyperBench, 932 have degree 2. Out of these 932 only 16 are obtained from synthetic queries. Furthermore, these hypergraphs are not necessarily simple and a significant number of them have high \ghw. Table~\ref{tab:deg} shows the number of degree 2 hypergraphs with $\ghw > k$ in detail. We see that of the 932 degree 2 hypergraphs, 649 are acyclic ($\ghw > 1$) and almost 400 have \ghw even higher than 5.

In summary, this suggests that degree 2 hypergraphs with non-trivial \ghw occur naturally in a variety of applications. This may also motivate the study of dilutions to jigsaws as a tool for determining \ghw or as a factor in solving degree 2 queries with high width.

\begin{table}[h]
  \centering
  \caption{Number of Degree 2 Hypergraphs in HyperBench with $\ghw > k$}
  \label{tab:deg}
  \begin{tabular*}{0.25\linewidth}[h]{c|c}
    $k$ & amount \\
    \hline 
    1 & 649 \\
    2 & 575 \\
    3 & 506 \\
    4 & 452 \\
    5 & 389
  \end{tabular*}
\end{table}

\section{Additional Details for Section~\ref{sec:dilution}}

\begin{proof}[Proof of Statement (3), Lemma~\ref{lem:dilut}]
  We will only argue that for any hypergraph $H$, merging all incident
  edges $I_v$ for a vertex $v$ by replacing all edges $I_v$ by a
  single new edge $e_v = \bigcup I_v \setminus \{v\}$ can not increase
  \ghw. Let us refer to the new hypergraph after the merging as
  $H'$. For the other operations the fact that \ghw only decreases is
  well known (see e.g.,~\cite{JACM21}). The full statement thus follows from proving this case.

  Let \tdecomp be a tree decomposition with minimal \ghw $k$ for $H$
  and associate a $\lambda_u$ to every $u\in T$ that describes a
  minimal edge cover in $H$ of each bag.  We will now derive new
  labels $(\lambda'_u)_{u \in T}$ such that for every $u\in T$ we have
  $|\lambda'_u|\leq |\lambda_u$ and $\lambda'_u$ is a set cover of
  $B_u \setminus \{v\}$ in $H'$.  We will then adapt the bags, such
  that at least one of them also covers the new edge $e_v$.

  For the appropriate new covers it is enough to set
  \[
    \lambda'_u =
    \begin{cases}
      (\lambda_u \setminus I_v) \cup \{e_v\} & \text{if } I_v \cap \lambda_u \neq \emptyset \\
      \lambda_u & \text{otherwise}
    \end{cases}
  \]
  By definition $e_v$ covers the same vertices as all edges in $I_v$ together (in $H'$).

  We now move on to defining the bags $B'_u$ of the new decomposition for $H'$.
  Let $T_v$ be the subtree  $\{u \in T \mid v\in B_u\}$. Then our new bags are defined as follows for every $u \in T$.
  \[
    B'_u =
    \begin{cases}
      B_u \cup e_v & \text{if } u \in T_v \\
      B_u & \text{otherwise}
    \end{cases}
  \]
  It is easy to see that all edges of $H'$ are now contained in some
  bag $B'_u$. The unchanged edges are still present in the same bag as before,
  and $e_v$ is in at least one bag, since $T_v$ can not be empty.
  
  To verify connectedness of the newly constructed decomposition it is enough to observe that every edge in $I_v$ will occur fully in some bag in $T_v$. This is because every must be fully covered by at least one bag. Since all edges in $I_v$ contain the vertex $v$, this must happen somewhere in $T_v$. With the updated bags, every vertex $w \in e_v$ now occurs in the union of $T_w$ and $T_v$. By the observation on $I_v$ above, $T_w \cap T_v \neq \emptyset$ and thus their union is again connected.
  
  What is left is to observe that $B'_u \subseteq \bigcup \lambda'_u$ for every node $u$.
  Observe that $v \in B_u$ only if $I_v \cap \lambda_u\neq\emptyset$ since the edges in $I_v$ are the only ones that contain $v$. Thus, we also have $e_v \subseteq B'_u$ only if $e_v \in \lambda'_u$. All unchanged edges are clearly still covered the same as in the original decomposition, as noted above.

  Hence, we see that $\left< T, B'_u\right>$ is a tree decomposition
  with \ghw at most $k$, since $\lambda'_u$ is a witness of set covers
  with at most $k$ elements for each bag.
\end{proof}

\begin{proof}[Proof of Theorem~\ref{thm:hmred}]
  Let $q$ be a CQ with hypergraph $M_q$ in $\mathcal{M}$ and $D_q$ a
  database with the same schema as $q$. In particular, we assume
  w.l.o.g. that $q$ has no self-joins. If it did we could reduce to a
  self-join free $q'$ with database $D'$ in polynomial time by
  splitting duplicate relation names in $q$ into new individual
  relation names where the relations in $D'$ are direct copies of the
  respective original relation in $D_q$.  The hypergraph of such a
  $q'$ would be the same as the hypergraph of $q$.

  By enumeration of $\mathcal{H}$, find a hypergraph $H$ such that
  $H$ dilutes to $M_q$  and let $W = (w_1,\dots,w_\ell)$
  be a dilution sequence from $H$ to $M_q$.  Note that $H$ and $W$
  depend on $M_q$, i.e., the parameter of the problem. We will now
  show that by traversing $W$ in reverse, we can construct (in
  fixed-parameter polynomial time) a query $p$ such that
  $\pi_{\vars{q}}(p(D_p)) = q(D_q)$, and
  the hypergraph of $p$ is $H$.

  To do so, we will show for each dilution operation $w_i$ -- 
  that produces hypergraph $H_{i}$ from $H_{i-1}$ -- how the query
  $q_{i-1}$ and database $D_{i-1}$ for $H_{i-1}$ can be transformed
  into an equivalent instance $q_i,D_i$ for $H_{i}$ with $\pi_{\vars(q_{i-1})}(q_i(D_i)) = q_{i-1}(D_{i-1})$.
  Thus, ultimately
  we can reduce from a query for $H_\ell=M_q$ to a query for
  $H_0=H$. We also argue for each operation that
  $\norm{D_{i-1}} \leq f(M_q)\norm{D_i}$, from which it will become apparent
  that this is indeed an fpt-reduction.

  It will be convenient to observe that the degree never increases along a dilution sequence, i.e., for all $1 \leq i \leq \ell$ it holds that \(\degree(H_i) \leq \degree(H_{i-1})\). The observation is easy to verify directly from Definition~\ref{def:dilut}.
  The reduction introduces new constants that will serve to link relations via functional dependence on the new constant.
  For this purpose, consider the new constants $(\const_i)_{i\geq 0}$ that do not occur in $D_q$. The final reduction will at most as many constants as the maximum number of tuples in a relation in $D_q$.
  \paragraph{ $w_i$ deletes a vertex $v$ from $H_{i-1}$}
  While the basic principle of this direction is simple, there are
  some technicalities that require a certain amount of care. In
  particular, deleting vertex $v$ can make two edges the same. Hence,
  reversing the operation is not as straightforward as the
  deletion. Fortunately, even a very direct approach will be enough
  for our purposes.
  
  Let $E_v$ be the edges in $H_{i-1}$ that are incident to $v$. For
  every edge $e \in E_v$, fix a $\pre(e) \in E(H_i)$ such that
  $\pre(e) \cup \{v\} = e$. Let $R_{\pre(e)}(\bar{x})$ be the atom in
  $q_i$ that corresponds to edge $\pre(e)$ in the hypergraph. Then,
  for each edge $e$ in $E_v$, create a new atom $S_e(\bar{x}, v)$ in
  $q_i$ where $\bar{x}$ are the arguments of atom
  $R_{\pre(e)}(\bar{x})$ and let
  \[
    S_e^{D_{i-1}} = R_{\pre(e)}^{D_i} \times \{(\const_0)\}
  \]
  where the product is interpreted as in relational algebra. The rest
  of $q_i$ and $D_i$ is made up of direct copies of those
  atoms/relations that correspond to edges that are in both $H_{i-1}$
  and $H_i$. Since all values in all tuples in the position of the newly introduced joins over $v$ are the same, it is straightforward to observe that $\pi_{\vars(q_{i-1})}(q_i(D_i)) = q_{i-1}(D_{i-1})$.

  Let us consider how the size of $D_{i-1}$ is related to the size of $D_i$.
  We create at most $\degree(v)$ new relations, where each relation is a relation from $D_{i-1}$ with each tuple extended by a constant.
  Thus, the representation of such a new relation of $S_e$ increases over the corresponding $R_{\pre(e)}$ only by some constant factor.
  Hence, overall at most $O(\degree(v) \norm{D_i})$ space (and time) is required to create the new relations. At most the whole previous database is kept, adding at most $\norm{D_i}$ size to the new $D_{i-1}$. Since degree never increases along dilution sequences, $\degree(v) \leq \degree(H)$ and we arrive at our bound of
  \[
    \norm{D_{i-1}} = O( \degree(H) \cdot \norm{D_{i}})
  \]
  \paragraph{ $w_i$ replaces the incident edges $E$ of vertex $v$ in $H_{i-1}$, by a new edge $e=\bigcup E \setminus \{v\}$ in $H_i$}
  Let $e_1, \dots, e_n$ be the edges that make up the set $E$.
  Let $R(\bar{e})$ be the atom corresponding to edge $e$ in $H_i$. In $q_{i-1}$ we replace $R(\bar{e})$ by new atoms $R_j(\bar{e_j})$ for every $j \in [n]$. Let  $v$  always be in the last position of the new atoms.
  To define the new relations, suppose $R^{D_i}$. Let $R'$ be $R^{D_i}$ extended by a new attribute $v$, with every tuple extended by a distinct $\const_i$ ($i \leq |R^{D_i}|$) in the new position. Let the new relations in $D_{i-1}$ for the new atoms $R_j$ in $q_{i-1}$ be
  \(
    R_j^{D_{i-1}} = \pi_{e_j}(R')
    \). Again everything except $R_e$ is copied directly from $q_i, D_i$.
    Since every tuple in $R'$ has a distinct $\const_i$ value for attribute $v$, $R'$ and, in consequence, every $R_j$ is functionally dependant on $v$. Since everything else in $q_i$ and $D_i$ remains unchanged we again have $\pi_{\vars(q_{i-1})}(q_i(D_i)) = q_{i-1}(D_{i-1})$.
    
  Clearly, the database can increase in size by no more than if we just copied $R'$ fully $n$ times. Again we see that $n\leq \degree(H_{i-1}) \leq \degree(H)$ and the detailed argument follows the same steps as in the vertex deletion case above.
  \[
    \norm{D_{i-1}} \leq  c \degree(H) \norm{D_{i}}
  \]
  \paragraph{ $w_i$ deletes a subedges $f \subset e$  from $H_{i-1}$}
  In this case it is enough to add a new $R_f(\bar{f})$ to $q_i$ to obtain $q_{i-1}$. The relation is naturally
  \(
    R_f^{D_{i-1}} = \pi_f(R_e^{D_i})
  \)
  and we have the following bound on size of the new database
  \(
    \norm{D_{i-1}} \leq  2 \norm{D_{i}}
  \). It is straightforward to verify that $q_i(D_i)=q_{i-1}(D_{i-1})$.
  \paragraph{Putting it all together.}
  We have shown how to reduce $q_\ell = q$ to $q_o =p$. The
  computational effort in each step from $i$ to $i-1$ consists only of extending relations by one attribute, copying a single relation, or projection, and is feasible in $O(\degree(H)(\norm{q_i}+\norm{D_i}))$ time. From the bounds on the database size derived
  for each operation we can deduce the following bound for the final
  database $D_p = D_0$
  \[\norm{D_p} = c\,\degree(H))^\ell \norm{D_q}\]
  Since we introduce no self-joins or duplicate variables in the same atom, the size of the final query depends only on the size of $H$.
  Recall, $H$ and $W$, and thus
  also $\ell$, depend only on the parameter $M_q$.   The described process thus
  reduces $q,D_q$ to $p,D_p$ in $f(M_q) (\norm{D_q})$ time, such that
  $\pi_{\vars(q)}(p(D_p)) = q(D_q)$.
\end{proof}

Before we show the \np-completeness, we first show the opposite of Lemma~\ref{lem:shmin} as its own
statement, and then observe \np-hardness of deciding hypergraph
dilutions as consequence of the two lemmas put together.
\begin{lemma}
  \label{lem:shminrev}
  Let $G$ be a connected graph and let $H$ a degree 2 hypergraph. If
  $G^d$ is a hypergraph dilution of $H$, then $G$ is a minor of $H^d$.
\end{lemma}
\begin{proof}
  Suppose $G^d$ is a dilution of $H$. We will construct an appropriate
  minor map $\mu \colon V(G) \to 2^{E(H)}$ from $G$ into $H^d$.
  
  For this purpose, suppose we keep track of labels $L(e)$ for the edges of the hypergraphs the dilution process. We set $L(e) = \{e\}$ initially and the labels are then updated as follows, depending on operation.
  When deleting a vertex collapses multiple edges $e_1,\dots,e_\ell$ into one edge $e_0$ we set $L(e_0)=\bigcup_{i=0}^\ell L(e_i)$ and copy the other labels unchanged. When deleting a subedge $e_1 \subset e_0$ we set $L(e_0) = L(e_1) \cup L(e_0)$ and copy any other labels unchanged. Finally, when merging edges $I_v$ over a vertex $v$, we set the label of the new edge $e_v$ as $L(e_v) = \bigcup_{e \in I_v} L(e)$.

  After dilution from $H$ to $G^d$, we then every edge of $G_v$ associated with a label which is a set of edges of $H$. Since $E(G^d)=V(G)$, $L$ is thus a function $V(G) \to 2^{E(G)}$. We claim that $L$ is a minor map, i.e., that every image set $L(e)$ is connected in $H^d$ and any two $L(e_1),L(e_2)$ are disjoint if $e_1\neq e_2$.

  We first observe the disjointness of any two labels in $G^d$.  Note that by construction, all labels are trivially disjoint in $H$. In every step, every label is either copied unchanged, or multiple labels are combined into a single label. Since the individual parts of this combined label are disjoint with all unchanged labels, so is the combined label.

  Connectedness of a set of edges implies connectedness of the respective vertices in the dual hypergraph. Hence, connectedness of a image of the minor map in $H^d$ follows directly from the connectedness of any $L(e)$ in $H$ in $G^d$. For the merging and subedge deletion operations it is straightforward to see that connectedness is preserved in the construction of the labels.
 When deleting a vertex, observe that multiple edges collapse into one only if their only difference was vertex $v$ and they are the same otherwise (hence actually $\ell \leq 1$ in the case above). Since they are the same otherwise they are connected via at least one vertex that is not $v$ (they can not both contain only $v$). Hence, $L$ is a minor map from $G$ into (and actually onto) $H^d$.
\end{proof}

\begin{proof}[Proof of Theorem~\ref{thm:npdilt}]
  Recall, it is known to be \np-complete to decide whether a graph $G$ is a minor of
  graph $F$~\cite{DBLP:books/fm/GareyJ79}. We prove \np-hardness of
  our problem by reduction from graph minor checking.  By
  Lemmas~\ref{lem:shmin} and~\ref{lem:shminrev} we have that $G$ is a
  minor of hypergraph $H^d$ if and only if $G^d$ is a hypergraph
  dilution of $H$. The desired reduction then follows from setting
  $H=F^d$ and observing that the dual of graph $F$ always has degree
  at most 2.

  \np-membership follows from the observation that hypergraph
  dilutions are, in a sense, \emph{monotonically decreasing}. That is,
  if $H'$ dilutes to $H$, then $|V(H)|\leq |V(H')|$,
  $|E(H)|\leq|E(H')|$, and at least one of the inequalities is strict. Hence, if $H$ is a hypergraph dilution of $H'$, then there is a linear length dilution sequence from $H'$ to $H$. Hence, a linear size guess of a dilution sequence leads to an \np algorithm for the problem.
\end{proof}

\section{Additional Details for Section~\ref{sec:forbid}}

\begin{proof}[Proof of Lemma~\ref{lem:dualtw}]
  Let $\left< T, (D_u)_{u \in T} \right>$ be a  tree
  decomposition of $H^d$ with width $k$.  We construct a generalised hypertree decomposition (GHD
  $\left< T, (B_u)_{u \in T}, (\lambda_u)_{u\in T}, \right>$ for $H$
  by taking for every node $u$ in $T$, $\lambda_u = D_u$ and
  $B_u = \bigcup \lambda_u$ (note that the elements of $D_u$ are edges in $H$). Recall, a GHD is a tree decomposition with an additional labelling $(\lambda_u)_{u\in T}$ that describes an explicit edge cover for each bag. 

  It is not difficult to verify that this is indeed a GHD of width $k+1$ of $H$.
  To do so we have to argue two properties (the width is trivial). First, that for every $e \in E(H)$, there is a node $u$ such that $e \subseteq B_u$, and second that the connectedness condition holds.

  For the first property, consider an arbitrary $e \in E(H)$ as a vertex in $H^d$. Then, there is some node $u$ such that $e \in D_u$, since we have a tree decomposition of $H^d$. Then, also $e \in \lambda_u$, and in consequence $e \subseteq B_u = \bigcup \lambda_u$.

  For connectedness, consider an arbitrary vertex $v \in V(H)$.  
  Let $f_v \in E(H^d)$ be the edge corresponding to $v$ in the dual. Recall, the elements of $f_v = \{e_1,\dots, e_n\}$ correspond to the edges incident to $v$ in $H$. Let $u$ be a node in $T$ such that $f_v \subseteq D_u$.
  Then, by connectedness of the TD, the subtrees $T_{e_i} = \{e_i \in D_u \mid u \in T\}$ for $e_i \in f_v$ are each connected and all contain the node $u$. Hence, also $T_v = \bigcup_{e_i \in f_v} T_{e_i}$ is connected. By our definition of the bags in the GHD, $v$ occurs exactly in nodes that have an $e_i \in f_v$ in their $\lambda$ label, i.e., in the nodes of $T_v$. Thus, we see that connectedness holds for every vertex in the constructed GHD.
\end{proof}

\section{Proofs for Section~\ref{sec:highdeg}}

We will again show our structural result on hypergraphs via graph minors in the dual. In particular, we first extend the notion of a graph minor to hypergraphs in a way that takes hyperedges into account. We will call such minors \emph{expressive minors}. We show how such expressive grid minors are related to normal grid minors via the rank of the hypergraph. From there we will then see that for classes of bounded degree, the dual hypergraphs also always contain large expressive grid minors if their treewidth is unbounded.

\begin{definition}[Expressive Minor Map]
  \label{def:expr}
  Let $G$ be a graph and $H$ a hypergraph. We say that a mapping $\mu \colon V(G) \to 2^{V(H)}$ is an \emph{expressive minor map} from $G$ into $H$ if $\mu$ is a minor map from $G$ onto $H$ and there exists a mapping $\rho \colon E(G) \to E(H)$ such that:
  \begin{enumerate}
  \item $\rho$ is injective, i.e., no two edges in $G$ are mapped to the same edge in $H$,
  \item for any edge $e=\{u,v\}$ in $G$, $\rho(e)$ intersects both $\mu(u)$ and $\mu(v)$,
  \item and for any two incident edges $e_1, e_2$  with $v\in e_1 \cap e_2$ in $G$, there is a path from $\rho(e_1)$ to $\rho(e_2)$ that uses only vertices in $\mu(v)$ and no edge in $\rho(E(G))$ (except for $\rho(e_1)$ and $\rho(e_2)$ as the start/end).  
  \end{enumerate}
   We say $G$ is an \emph{expressive minor} of $H$ if there exists an expressive minor map from $G$ into $H$.
\end{definition}
Note that if $H$ is a simple graph, i.e., if the rank of all edges is 2, then every minor is an expressive

Previous work considered graph minors into the Gaifman graph of hypergraphs. However, since every hyperedge of rank $r$ becomes an $r$-clique in the Gaifman graph, complex parts of the graph can be mapped into single large enough hyperedges, creating a mismatch in structure between a hypergraph and its graph minor.
The aim of expressive minors is to retain more structure from graph minors in hypergraphs by also imposing requirements on the edge structure of the hypergraph. Despite the additional restrictions, we can still observe the existence of large \emph{expressive} grid minors in the following sense.

\begin{figure*}[t]
  \centering
  \begin{subfigure}[b]{0.7\textwidth}
    \centering
    \includegraphics[width=\textwidth]{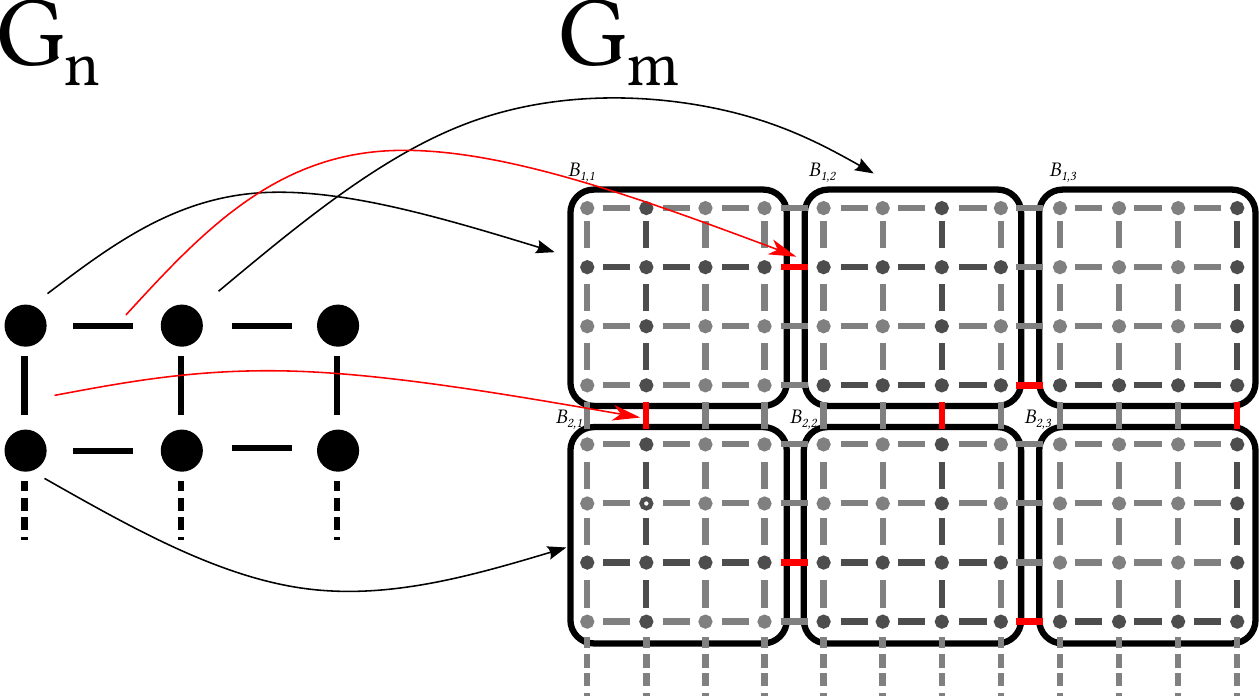}
    \caption{Illustration of intended construction. The large grid $G_m$ is grouped into $\nxn$ blocks, which correspond to single vertices in the $\nxn$ grid $G_n$. For each connection between blocks a single edge is fixed (in red), such that its row/column inside the block will not touch (after mapping to $H$ via $\mu_m$, not pictured) any other edge in $\rho$. }
    \label{fig:ex.fullblock}
  \end{subfigure}
  \hspace{5em}
  \begin{subfigure}[b]{0.2\textwidth}
    \centering
    \includegraphics[width=\textwidth]{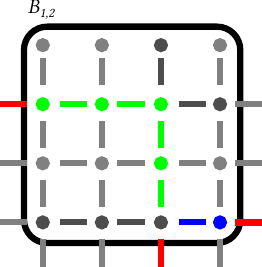}
    \vspace{2em}
    \caption{Paths inside block $B_{1,2}$ of Figure~\ref{fig:ex.fullblock} using the rows/columnd of fixed edges from $G_m$, corresponding to the paths between marked edges in $\rho$.}
    \label{fig:ex.block12}
  \end{subfigure}
  \hfill
\end{figure*}

\begin{lemma}
  \label{lem:expminorexcl}
  Let $H$ be a reduced hypergraph and let $m = 4 \rank(H)^5 n^5$. If the $m\times{}m$-grid is a minor of the primal graph of $H$, then the $\nxn$-grid is an expressive minor of $H$.
\end{lemma}
\begin{proof}
  Let $a=rank(H)$, let $G_m$ be an $m \times{}m$-grid and let $\mu_m$ be a
  minor map from $G$ \emph{onto} $H$. Now consider the partition of $G$ into
  $n\times{}n$ blocks of dimension $(4a^5n^4)\times(4a^5n^4)$. That is,
  for $k\in[n], l\in [n]$, we consider blocks $B_{k,\ell}$ which
  contain the vertices $v_{i,j}$ of the grid where
  $(k-1)\cdot 4a^5n^4\leq i \leq k\cdot 4a^5n^4$ and
  $(\ell-1)\cdot 4a^5n^4\leq j \leq \ell\cdot 4a^5n^4$.

  Let $G_n$ now be an $\nxn$-grid and define
  $\mu_n \colon V(G) \to 2^{V(H)}$ as
  $\mu_n(v_{k, \ell}) = \bigcup_{u \in B_{k,\ell}}\mu_m(u)$ for all
  $k\in [n], \ell\in[n]$. It is straightforward to observe that $G_n$
  is a minor of $G_m$ and therefore $G_n$ is also a minor of the Gaifman graph of  $H$.
  Since vertices in $G_n$ correspond one-to-one to blocks in $G_m$ we also write $B_v$ for the block in $G_m$ corresponding to vertex $v$ in $G_n$.
  We now construct an appropriate mapping
  $\rho \colon E(G) \to E(H)$ to show that $\mu_n$ is indeed an
  expressive minor. Below, we refer to the edges in $E(H)$ that are mapped to by $\rho$ as \emph{marked} edges.

  The most challenging part in constructing $\rho$ is to respect the third condition of Definition~\ref{def:expr}. From the grid minor we have knowledge of paths between different edges, in particular if there is a path from edges $u$ to $v$ in $G_m$, then there is a path from any edge touching $\mu_m(u)$ to any edge touching $\mu_m(v)$ in $H$ since the images of a minor map are always connected and adjacency in $G_m$  implies adjacency of the respective images in $H$ (the same clearly also holds for $G_n$ and $\mu_n$).  However, avoiding marked edges complicates the situation, as not every path in the grid necessarily has an analogue in $H$ without using marked edges, since a single hyperedge in $H$ can cover the connection between many vertices. %

  We first give an outline of the overall strategy of our construction. For any
  two adjacent blocks in $G_m$, we will fix one of the
  edges that connect the two blocks in $G_m$ and mark an edge in $H$ that covers the image of that edge. When we
  make a choice for horizontally adjacent blocks, we want to chose the
  edge in the $i$-th row of the block only if (the image under $\mu_m$ of) no  vertex in the
  $i$-th row of the block is contained in some marked edge (except for the marked edges that connect the block to the adjacent block).   For
  vertical connections our goal is same but for columns instead of
  rows. For block $B$, let $c(B) \subseteq E(H)$ be the edges marked in this way that connect  $\mu(B)$ to the image of its adjacent blocks. The high-level idea is illustrated in Figure~\ref{fig:ex.fullblock} with edges corresponding to an edge in $c(B)$ marked in red and block groupings marked in black.

  Since all pairs of rows and columns in a block intersect we then get that for any block $B$, there is a path between each pair of edges in $c(B)$ using only vertices that are in no marked edges (see Figure~\ref{fig:ex.block12}. Hence, there will also be a path in $H$ using no marked edges and only vertices from inside the block $B$, i.e., from the image $\mu_n(v)$ for some $v \in V(G_n)$.

  \begin{claim}
    \label{claim:blockedges}
    Let $B$, $B'$ be two adjacent blocks in $G_m$. Then there are at least $4a^4n^4$ distinct edges in $E(H)$ that touch both $\mu_m(B)$ and $\mu_m(B')$.
  \end{claim}
  \begin{claimproof}
    For any two adjacent blocks $B, B'$ in $G_m$, there are
    $4a^5n^4$ edges in $G_m$ that touch both blocks.  Thus, also the Gaifman graph of $H$ has at least that many edges that touch $\mu_m(B)$ and $\mu_m(B')$.
    Every edge of
    $H$ contains at most $a$ vertices and therefore also at most $a$ elements of $\mu_m(V(G_m))$. %
    That is, a single edge in $E(H)$ can only touch the image of at most $a$ vertices on the ``boundary'' (the rows and columns with the lowest and highest index in the block) of every block, i.e., those vertices that are adjacent to another block. 
    It follows that there are at least $4a^4n^4$ distinct edges in $E(H)$ that touch both
    $\mu_m(B)$ and $\mu_m(B')$. That
    is, for each edge in $G_n$ there are at least $4a^4n^4$ possible
    distinct choices to map to in $\rho$ that satisfy the second condition of Definition~\ref{def:expr}. %
  \end{claimproof}

  \bigskip

  In the following we will restrict the columns in blocks from which we choose ``vertical'' edges that connect to an adjacent block ``below'' or ``above'', and the rows from which we choose ``horizontal'' edges that connect to adjacent blocks to the ``left'' or ``right''. To avoid the cumbersome distinction we always refer only to the \emph{line of an edge in a block}, with the understanding that for ``vertical'' edges this refers to the column and for ``horizontal'' edges this refers to the row.

  We now argue that the $4a^4n^4$ possible choices per edge in $G_n$ are enough to establish our $\rho$ as intended.
  We give a procedure that constructs such a $\rho$. Fix some ordering $O = (e_1,e_2,\dots,e_\ell)$ on the edges of $G_n$. In a first phase, in order of $O$, for $e_i = \{v,u\}$ fix a tuple $\beta_i = (F_i, \rho_i)$, where $F_i \subseteq E(G_m)$ is a set of $2n^2a$ edges that touch $B_v$ and $B_u$, and $\rho_i \colon F_i \to E(H)$ is an injective map such that for every $f \in F_i$ it holds that $\mu_m(f) \subseteq \rho_i(f)$. Let $R(\beta_i) = \{w\in V(G_m) \mid \exists f\in F_i. \mu_m(w) \cap \rho_i(f) \neq \emptyset \}$ and note that $|V(\beta_i)| \leq 2n^2a^2$ since there are at most that many vertices in the fixed edges of $H$ and no vertex is in two images of $\mu_m$. Additionally, we require the following condition on $F_i$: for every $f\in F_i$,
  for the set of vertices $V_{v,f}$ in the line of $f$ in $B_v$, it holds that $V_{v,f} \cap \bigcup_{j < i} R(\beta_j) = \emptyset$ (and analogously for $B_u$). That is, we only fix edges that lie on rows and columns that are not touched (in the image) by any edge of $H$ that was fixed in a previous step. Note that this condition also implies that $\rho_i$ can not map to any hyperedge in the image of a previous $\rho_j$ for $j<i$.
  Let us refer to such a $\beta_i$ as a \emph{bundle} for the $i$-th edge in $O$.

  We first argue that such a sequence of bundles always exists before moving on to the second phase of the procedure.
  At step $i$ in the procedure, $2(i-1)\ n^2a$  distinct hyperedges have been fixed already, and thus $\bigcup_{j<i} R(\beta_j)$ contains at most $2(i-1)\ n^2a^2$ vertices of $G_m$. Hence, at most that many rows and columns are ``blocked'' (twice the number of vertices since a vertex occurs in a row and a column of each block) by previous choices and $4a^5n^4 - 4(i-1) n^2a^2$ edges that connect the two blocks of step $i$ in $G_m$ are feasible and thus also at least $4a^4n^4 - 4(i-1) n^2a^2$ edges in $E(H)$ to fix in $\rho_i$.  Since $\ell < n^2$, this leaves enough feasible choices for $\beta_i$ and $\rho_i$ to fix in every step of the first phase.
  By Claim~\ref{claim:blockedges} there are 

  This first phase establishes a possible list of choices for the mapping $\rho$ for every edge in $G_n$ with some limited guarantees on the paths between them.
  What is left, is to now filter down the bundles to
  single edges in a way that satisfies the third condition of Definition~\ref{def:expr}. This can
  be achieved by repeating the basic idea of the first phase in reverse direction of $O$, but with choices restricted to the already fixed bundles. Note that in step $i$ of the reverse order corresponds to step $\ell-i+1$ in the forward iteration over $O$ in the first phase.
  We iterate in reverse order of $O$. In step $i$, we select a single $f'_i \in F_{\ell-i+1}$ and fix $\rho(e_{\ell-i+1}) = \rho_{\ell-i+1}(f'_i)$. Let $R_i = \{w \in V(G_m) \mid \mu_m(w) \in \rho_{\ell-i+1}(f'_i)\}$,
  the choice of $f'_i$ (incident to blocks $B_v, B_u$) shall satisfy the following property: 
  for the set of vertices $V_{v,f'_i}$ in the line of $f'_i$ in $B_v$, it holds that $V_{v,f'_i} \cap \bigcup_{j > \ell-i+1} R_{j} = \emptyset$. That is, $f'_i$ is chosen such that no vertex in the same line on the block has its image (w.r.t. $\mu_m$) in an edge of $H$ fixed in a previous step of the second phase.

  In this phase, after step $i$ we have fixed $i$ choices and thus $\bigcup_{j < i} R_{\ell-j+1}$ contains at most $ia$ vertices, leaving at least $2n^2a - 2ia$ feasible choices. As before we see that with $\ell < n^2$ steps in total, there is at least one feasible choice at every step of the second phase. Furthermore, in step $i$, it is easy to observe that there is exactly one edge $e$ in $G_n$ such that $\rho(e_n)$ intersects $V_{v,f'_i}$, namely $e = e_{\ell-i+1}$. By the first phase, we have that for every $0 \leq j < \ell-i+1$ the line $V_{v,f'_i}$ contains no vertices that map into an edge in the image of $\rho_j$ and thus particularly not into $\rho(e_j)$. From the condition in the second phase, it also holds that no vertex in $V_{v,f'_i}$ maps into a vertex in any $\rho(e_j)$  for $j > \ell-i+1$. As argued above, this then implies that $\rho$ satisfies the third condition of Definition~\ref{def:expr}. 
  Furthermore, the resulting $\rho$ is clearly injective and for every $e = \{u,v\}$ in $G_n$, $\rho(e)$ touches both $B_u$ and $B_v$. Hence, $\mu_n$ is an expressive minor from the $\nxn$-grid $G_n$ into $H$.
\end{proof}

\begin{theorem}
  \label{thm:excexpr}
  There exists a function $f \colon \mathbb{N}^2 \to \mathbb{N} $ with the following property. For every
  $n\geq 1$, every bounded rank hypergraph $H$ with $\tw(H) > f(n, \rank(H))$ contains the $\nxn$-grid as an expressive minor.
\end{theorem}

\begin{lemma}
  \label{lem:exp2qj}
  Let $H$ be a reduced hypergraph such that the $\nxn$-grid is an expressive minor of $H^d$.
  Then there exists an $\nxn$-pre-jigsaw $H'$ such that $H$ dilutes to $H'$.
\end{lemma}
\begin{proof}
  First, we assume w.l.o.g. that $H$ has no isolated vertices, no empty
  edges, and no duplicate vertex types. It is straightforward to
  observe that such vertices and edges can not contribute to the
  minor mapping of $G$ into $H^d$ in any meaningful way. Note that these
  assumptions are only made in the argument for sake of simplicity and
  the statement of the lemma still holds in full generality since the
  assumed properties can always be enforced through a simple
  dilution sequence (cf. Lemma~\ref{lem:redhyper}).

  Let $G$ be the $\nxn$-grid and let $J$ be the corresponding $\nxn$-jigsaw such that $J = G^d$.
  Let $\mu  \colon V(G) \to 2^{V(H^d)}$ be an expressive minor
  map from $G$ onto $H^d$.
  Let $\rho \colon E(G) \to E(H^d)$ be the mapping as in
  Definition~\ref{def:expr}.

  We will now consider the mappings $\mu$ and $\rho$ from the perspective of $G^d$ and $H$ instead.
  That is, let $\pi \colon V(G^d)\to V(H)$
  such that $\pi(x) = \rho(x)$, and analogously let $o$ be the respective
  dualisation $E(G^d) \to 2^{E(H)}$ of $\mu$.  Alternatively, we see
  that $\pi$ is a mapping $V(J) \to V(H)$, and $o$ is of the form
  $E(J) \to 2^{E(H)}$.

  Since $\mu$ is a minor, we have that for every two distinct vertices $v,u \in V(G)$, $\mu(v) \cap \mu(u) = \emptyset$, and thus also for any two distinct $e,f \in E(J)$ that $o(e) \cap o(f) = \emptyset$. Since $\mu$ is onto, also every edge in $H$ is in some image of $o(e)$ for $e \in E(J)$ (since every vertex of $H^d$ is in some $\mu(v)$ for $v \in V(G)$.

  Finally, for any two vertices $u,v$ in edge $e$ of $J$, there are edges $f_v, f_u$ with a common vertex $w_e$ in $G$. Since $\mu$ is an expressive minor (witnessed by $\rho$), there is a path from $\rho(f_v)$ to $\rho(f_u)$ using only vertices in $\mu(w_e)$ and no other edges in the image of $\rho$. As noted before, it is easy to see that there also exists a path from $\pi(v)$ to $\pi(u)$ using only edges in $o(e)$ and no vertices in the image of $\pi$. Let $P_{u,v}$ be such a path.

  That is, $H$ already satisfies the first three properties of a pre-jigsaw. We show that the fourth property can always be enforced via dilution while preserving the other three.
  In particular, let $C$ be the set of all vertices that occur in some path $P_{u,v}$ for adjacent $u,v \in V(J)$ as fixed above. Recall that the paths are in $H$, and thus $C \subseteq V(H)$.
  
  To obtain $H'$ we delete all vertices that are neither in $C$ nor $\pi(V(J))$, as well as any empty edges that are created in the process.
  To adapt $o$ to edges of $H'$ we take
  \[
    o' \colon e \mapsto \{ f \cap V(P) \mid f \in o(e)\} \setminus \emptyset
  \]
  It is straightforward to verify that any $o'(e)$ is still connected by
  the paths connecting the pairs of vertices in $e$ since the each path $P_{u,v}$ is still a path in $H'$. The first two properties of pre-jigsaws are inherited directly from $H$.
\end{proof}

\begin{proof}[Proof of Theorem~\ref{thm:prej}]
  Let $H$ be a hypergraph with degree $d$. Let $f_d$ be the function from Theorem~\ref{thm:excexpr} with the second parameter fixed to $d$. Note that $H^d$ has rank $d$ and thus by~Theorem~\ref{thm:excexpr} contains a $\nxn$-grid minor. By Lemma~\ref{lem:exp2qj} $H$ then also dilutes to a pre-jigsaw.
\end{proof}

 \end{document}